\documentclass[10pt,a4paper]{article}

\usepackage{amsmath,amssymb,amsthm,mathrsfs}
\usepackage{stmaryrd}  
\usepackage{graphicx}
\usepackage{xcolor}
\usepackage{soul}
\usepackage{ulem}
\usepackage{fullpage}
\usepackage{bm}
\usepackage{cancel}
\usepackage{framed}
\definecolor{shadecolor}{rgb}{0.92,0.92,0.92}
\newlength{\defbaselineskip}
\newcommand{\setlinespacing}[1]%
{\setlength{\baselineskip}{#1 \defbaselineskip}}

\theoremstyle{plain}
\newtheorem{theorem}{Theorem}[section]

\newtheorem{proposition}[theorem]{Proposition}

\theoremstyle{definition}

\theoremstyle{remark}
\newtheorem{remark}[theorem]{Remark}

\numberwithin{equation}{section}

\newcommand{\bE}{\mathbb{E}}

\allowdisplaybreaks
\begin{document}
	\title{A Mean-Field Control Problem of Optimal Portfolio Liquidation with Semimartingale Strategies}

	\author{Guanxing Fu\thanks{The Hong Kong Polytechnic University, Department of Applied Mathematics, Hung Hom, Kowloon, Hong Kong. G. Fu’s research is supported by The Hong Kong RGC (ECS No.25215122) and NSFC Grant No. 12101523. } \quad \quad Ulrich Horst\thanks{Humboldt University Berlin, Department of Mathematics and School of Business and Economics, Unter den Linden 6, 10099 Berlin}\quad \quad Xiaonyu Xia\thanks{Wenzhou University, College of Mathematics and Physics, Wenzhou 325035, PR China; corresponding author. X. Xia’s research is supported by NSFC Grant No. 12101465.}}

	\maketitle
	
\vspace{-2.6mm}
\begin{abstract}
We consider a mean-field control problem with c\`adl\`ag semimartingale strategies arising in portfolio liquidation models with transient market impact and self-exciting order flow. We show that the value function depends on the state process only through its law, and that it is of linear-quadratic form and that its coefficients satisfy a coupled system of non-standard Riccati-type  equations. The Riccati equations are obtained heuristically by passing to the continuous-time limit from a sequence of discrete-time models. A sophisticated transformation shows that the system can be brought into standard Riccati form from which we deduce the existence of a global solution. Our analysis shows that the optimal strategy jumps only at the beginning and the end of the trading period.   		
		
\end{abstract}
	
	

\section{Introduction}
	Let $T \in (0,\infty)$ and let $W$ be a Brownian motion on a probability space $(\Omega, \cal F, \mathbb P)$ and $\mathbb F=(\mathcal F_t)_{t \in [0,T]}$ be the augmented Brownian filtration. In this paper we consider the mean-field stochastic control problem 
	\begin{equation*}
		\min_{Z \in  \mathscr A} \mathbb E\left[ \int_0^T\left( Y_{s-}\,dZ_s+\frac{\gamma_{2}}{2}\,d[Z]_s + \sigma_s d[Z,W]_s \right)+\int_0^T\lambda X^2_s\,ds\right]
	\end{equation*}
	subject to the state dynamics
	\begin{equation*}
		\left\{\begin{split}
			&~dX_s=-dZ_s\\
			&~dY_s=-\rho Y_s\,ds+\gamma_{1}d C_s+\gamma_{2}\,dZ_s+\sigma_s\,dW_s\\
			&~dC_s=-(\beta-\alpha)C_s ds+\alpha(\mathbb E[x_0]-\mathbb E[X_s]) \,ds\\
			&~X_{0-} = x_0; ~ Y_{0-}=C_{0-}=0; ~ X_T = 0
		\end{split}\right.
	\end{equation*}
	for $0 \leq s \leq T$ where the set of admissible controls is given by 
	\begin{equation*}
		\mathscr A:=\left\{Z: Z\textrm{ is a c\`adl\`ag } \mathbb F\textrm{ semimartingale with }{\mathbb E[[Z]_T]}<\infty	\textrm{ and }{\mathbb E\left[\sup_{0\leq t\leq T}Z^2_t\right]<\infty}	\right\}
	\end{equation*}
	
	Control problems of the above form arise in models of optimal portfolio liquidation with instantaneous and transient market impact and self-exciting order flow. In such models the process $(X_t)_{t \in [0,T]}$ describes the portfolio holdings (``inventory'') of a large investor, the terminal state constraint $X_T=0$ reflects the liquidation constraint, $(Z_t)_{t \in [0,T]}$ is the trading strategy, and $(Y_t)_{t \in [0,T]}$  specifies the transient impact of the investor's past trading on future transaction prices. We may think of $Y$ as an additional drift ``added'' to an unaffected (martingale) benchmark price process or a randomly fluctuating spread that increases linearly in order flow and recovers from past trading at a constant rate $\rho$. The process $(C_t)_{t \in [0,T]}$  can be viewed as describing the expected number of child orders resulting from the large investor's trading activity. The first term in the cost function describes the trading cost while the last term captures the market risk. These two terms are standard in the liquidation literature; see \cite{Alfonsi-2016, AJK-2014, Cartea-2018, FGHP-2018, HXZ-2020, PZ-2018} and references therein. When working with semimartingale strategies one needs to penalize the (co)variation of the semimartinagle strategy (with the driving Brownian motion); otherwise the optimization problem would not be well posed. This observation goes back at least to G\^arleanu and Pedersen \cite{GP}. This explains the second and third term in our cost function. Similar cost terms have also been considered in \cite{Ackermann2021, Hkiv, lorenz_schied}. 

\subsection{Portfolio liquidation models}
{The impact of limited liquidity on optimal trade execution has been extensively analyzed in the mathematical finance literature in recent years. The majority of the optimal portfolio liquidation literature allows for one of two possible price impacts. The first approach, pioneered by Bertsimas and Lo \cite{BL-1998} and Almgren and Chriss \cite{AC-2001}, divides the price impact in a purely temporary effect, which depends only on the present trading rate and does not influence future prices, and in a permanent effect, which influences the price depending only on the total volume that has been traded in the past. The temporary impact is typically assumed to be linear in the trading rate, leading to a quadratic term in the cost functional. The original modeling framework has been extended in various directions including general stochastic settings with and without model uncertainty and multi-player and mean-field-type models by many authors including Ankirchner et al.~\cite{AJK-2014}, Cartea and Jaimungal~\cite{cartea_jaimungal_2016}, Fu et al.~\cite{FGHP-2018,FHH,FHX1}, Gatheral and Schied~\cite{GatheralSchied11}, Graewe et al.~\cite{GHQ-2015}, Horst et al.~\cite{HXZ-2020}, Kruse and Popier~\cite{KP-2016} and Neuman and Vo{\ss}~\cite{neuman_voss}.
	
	A second approach, initiated by Obizhaeva and Wang \cite{OW-2013}, assumes that price impact is not permanent, but transient with the impact of past trades on current prices decaying over time. When impact is transient, one often allows for both absolutely continuous and singular trading strategies. When singular controls are admissible, optimal liquidation strategies usually comprise large block trades at the initial and terminal time. The work of Obizhaeva and Wang has been extended by Alfonsi et al. \cite{alfonsi_fruth_schied}, Chen et al. \cite{chen_kou_wang}, Fruth et al. \cite{FruthSchoenebornUrusov14}, Gatheral \cite{Gatheral2011}, Horst and Naujokat \cite{HN-2014}, and Predoiu et al. \cite{predoiu_shaikhet_shreve}, among others. 
	
	For single-player models Graewe and Horst \cite{GH-2017} and Horst and Xia \cite{HX1} combined instantaneous and transient impacts into a single model. Although only absolutely continuous trading strategies were admissible in both papers, numerical simulations reported in \cite{GH-2017} suggest that if all model parameters are deterministic constants, as the instantaneous impact parameter converges to zero the optimal portfolio process converges to the optimal solution in \cite{OW-2013} with two block trades and a constant trading rate. 
	
	The recent work of Horst and Kivman \cite{Hkiv} provides a rigorous convergence analysis within a Markovian factor model. They proved that in the stochastic setting, the optimal portfolio processes obtained in \cite{GH-2017} converge to a process of {\it infinite variation} with jumps as the instantaneous market impact parameter converges to zero and showed that the limiting portfolio process is of infinite variation and is optimal in a liquidation model with general {\it semimartingale execution strategies}. As such, their work provides a microscopic foundation for the use of semimartingale strategies in portfolio liquidation models. 
	
	Liquidation models with inventory processes of infinite variation and semi-martingale strategies were first considered by Lorenz and Schied \cite{lorenz_schied} and G\^arleanu and Pedersen \cite{GP}. Later, Becherer et al. \cite{becherer_bilarev_frentrup} considered a trading framework with generalized price impact and proved that the cost functional depends continuously on the trading strategy, considered in several topologies. More recently, Ackermann et al. \cite{Ackermann2021} considered a liquidation model with general RCLL semimartingale trading strategies. In a follow-up paper \cite{Ackermann2022} the same authors showed how to reduce liquidation problems with block-trades to standard LQ stochastic control problems. The main difficulty in \cite{Ackermann2021,Ackermann2022} arises from allowing a very general filtration while we consider a Brownian filtration but allow for self-exciting order flow which results in the said mean-field control problem. 
	
	Carmona and Webster \cite{carmona_webster} provide empirical evidence that inventories of large traders do indeed have a non-trivial quadratic variation component. For instance, for the Blackberry stock, they analyze the inventories of ``the three most active Blackberry traders'' on a particular day, namely CIBC World Markets Inc., Royal Bank of Canada Capital Markets, and TD Securities Inc. From their data, they ``suspect that RBC (resp.~TD Securities) were trading to acquire a long (resp.~short) position in Blackberry'' and found that the corresponding inventory processes were with infinite variation. More generally, they find that systematic tests ``on different days and other actively traded stocks give systematic rejections of this null hypothesis [quadratic variation of inventory being zero], with a $p$-value never greater than $10^{-5}$.'' 
	
	In this paper we consider a single player liquidation model with semi-martingale strategies and expectations feedback where the large investor expects his current trading activity to have an impact on future market dynamics. There are many reasons why (large) orders may have an impact on future price dynamics. For instance, extensive selling may diminish the pool of counterparties and/or generate herding effects where other market participants start selling in anticipation of further price decreases. Extensive selling may also attract predatory traders that employ front-running strategies; see Brunnermeier and Pedersen \cite{Brunnermeier-2005}, Carlin et al. \cite{Carlin2007} and Schied and Sch\"oneborn \cite{SS-2009} for an in-depth analysis of predatory trading.  
	
	Single and multi-player liquidation models with expectation feedback and absolutely continuous strategies {where $dZ_t = \xi_t dt$} have been considered by, e.g.~Cay\'e and Muhle-Karbe \cite{Caye-2016} and more recently Chen et al. \cite{CHT} and Fu et al. \cite{FHX1}. While Cay\'e and Muhle-Karbe \cite{Caye-2016} assume an exogenous augmentation of the large trader's order flow, due to e.g., predatory traders, Chen et al. \cite{CHT} and Fu et al. \cite{FHX1} assume that the market order dynamics follows a Hawkes process with exponential kernel, and that the large investor’s trading triggers an additional flow of child orders whose rate increases linearly in the investor’s expected traded volume. We retain the assumption of market order arrivals following Hawkes processes but allow for general semi-martingale execution strategies. This results in a novel mean field control problem. 

\subsection{Contributions to control problems}

	While the restriction to absolutely continuous controls is standard in much of the liquidation literature, the assumption seems restrictive; it is often made for mathematical convenience as the resulting control problem is much simpler to analyze. In fact, while abstract existence and characterization of solutions results can be obtained for models allowing for more general classes of admissible strategies (see \cite{HN-2014} and references therein), explicit solutions in stochastic settings are rarely available. Retaining the assumption that the large investor's trading activity triggers an absolutely continuous flow of child orders that increases linearly in his/her traded volume, we obtain explicit solutions for both the value function and the optimal strategy when allowing for general c\`adl\`ag semimartingale trading strategies. 

To the best of our knowledge, we are the first to address mean field control problems with c\`adl\`ag semimartingale strategies, which include mean field singular control problems as special cases. Mean field singular control problems have been considered by many authors, including \cite{FH-2017,pham:wei:20,Hafayed-2013,Oksendal-2017}. Among them, using a relaxed approach Fu and Horst \cite{FH-2017} established an existence of optimal control result for a general class of mean field singular control problems. Guo et al. \cite{pham:wei:20} established a novel It\^o's formula for the flow of measures on semimartingales; however, the examples provided in \cite{pham:wei:20} are mean field control problems with either regular or singular controls. Hafayed \cite{Hafayed-2013} established a maximum principle for a general class of mean field type singular control problems. Using a maximum principle approach, Hu et al. \cite{Oksendal-2017} studied mean field type singular control games arising in harvesting problems. 

Our paper also contributes to the literature on the characterization of non-Markovian singular control problems. One-dimensional models were studied in, e.g.   \cite{Bank-2005,Bank-Karoui-2004,Bank-Riedel-2001,lorenz_schied}. Multidimensional  problems are much more difficult to analyze, even in Markovian settings; see Dianetti and Ferrari \cite{Dianetti-Ferrari-2021}. For non-Markovian ones, refer to e.g. \cite{Ackermann2021,EMoreauP-2018}. In \cite{Ackermann2021}, Ackermann et al. solved a two-dimensional problem with random coefficients arising in optimal liquidation problems. In \cite{EMoreauP-2018}, Elie et al. studied a multidimensional path-dependent singular control problem arising in utility maximization problems. Our control problem is a three-dimensional non-Markovian one
with the non-Markovianity arising from a possibly non-Markovian volatility process $\sigma$ and the mean field term $\mathbb E[X]$. Moreover, our strategy is of infinite variation.

The main challenge when analyzing control problems with semimartingale strategies (with or without mean-field term) is that there are usually no canonical candidates for the value functions and/or the optimal strategies. Even the linear-quadratic case is difficult to analyze; although it is intuitive that the value function is of linear-quadratic form, the dynamics of the coefficients is a priori not clear. This calls for case-by-case approaches when analyzing such problems. We follow the approach taken in \cite{Ackermann2021, GP} and consider a sequence of discrete time models and then pass to a heuristic continuous time limit. Our approach suggests that the value function depends on the state process only through its distribution and that it is of linear quadratic form driven by three deterministic processes and a BSDE. 

It turns out that the driving processes follow a system of Riccati equations that does not satisfy the assumptions that are usually required to guarantee the existence and uniqueness of a solution. Our main mathematical contribution is to show - through a sophisticated transformation - that our system can be rewritten in terms of a more standard system that satisfies the assumption in Kohlmann and Tang \cite{Kohlmann2003} under a weak interaction condition that bounds the impact of the child order flow on the market dynamics.  Subsequently, we employ a non-standard verification argument that shows that the candidate optimal strategy is indeed optimal. The key idea is to rewrite the cost function as a sum of complete squares plus a correction term that turns out to be the value function. 

Our analysis shows that the optimal strategy jumps only at the beginning and the end of the trading period. This is consistent with earlier findings in \cite{GP, Hkiv, HN-2014,OW-2013}; in the absence of an external ``trigger'' there is no reason for the optimal strategy to jump, except at the terminal time to close the position and at the initial time.     

The remainder of this paper is structured as follows. The models, main results  and assumptions are stated in Section \ref{sec-main}. The wellposedness of the system of the Riccati equations that specify the candidate value function and strategy is established in Section \ref{sec-ric}. The verification argument is given in Section \ref{sec-ver}. Numerical simulations illustrating the nature of the optimal solution are given in Section \ref{sec-num}. Section \ref{sec-con} concludes. The heuristic derivation of the value function is postponed to an appendix.

\textbf{Notation.} For a deterministic function $\mathcal Y$, denote by $\dot{\mathcal Y}$ its derivative. For a stochastic process $\mathcal Y$ satisfying some SDE, we still use the same notation $\dot{\mathcal Y}$ to denote the drift of $\mathcal Y$. For a matrix (or vector) $\mathcal Y$, $\mathcal Y_{ij}$ (or $\mathcal Y_i$) denotes its $(i,j)$- (or $i$-) component. For a space $\mathcal D$, we denote by $L^\infty([0,T];\mathcal D)$  the space of all $\mathcal D$-valued bounded functions.  $C([0,T];\mathcal D)$ is the space of $\mathcal D$-valued continuous functions. For an integer $n$, we denote by $\mathbb S^n$ the space of symmetric $n\times n$ matrices. For a measure $\mu$ on ${\mathbb R}^n$ we denote by $\bar \mu$ the vector of expected values, and for a matrix  $A\in \mathbb R^{n\times n}$ we put
\[
\textnormal{Var}(\mu)(A):=\int_{\mathbb R^3}(x-\overline\mu )^\top A( x-\overline\mu  ) \mu(dx)=  \int_{\mathbb R^3}x^\top A x\mu(dx)-\overline\mu^\top A\overline\mu.  
\]
\section{Model and main results}\label{sec-main}

{In this section, we introduce a novel single player portfolio liquidation problem with semi-martingale execution strategies and expectations feedback. Our starting point is a novel ``semi-martingale extension'' of the portfolio liquidation model with instantaneous and persistent price impact analyzed in \cite{GH-2017} that allows the impact process to be driven by noise trading. We introduce this model in the next subsection before introducing an additional feedback term of mean-field type into the dynamics of the benchmark price process.  
	
	{We assume throughout that randomness is described by a Brownian motion $W$ defined on a filtered probability space $(\Omega, \cal F, \mathbb P)$ where $\mathbb F=(\mathcal F_t)_{t \in [0,T]}$ denotes the augmented Brownian filtration.   }
	
	\subsection{The benchmark model}
	
	We consider an investor that needs to unwind an initial of $x_0$ shares over a finite time horizon $[0,T]$ using general semi-martingale execution strategies $Z$. The large investor’s stochastic control problem at time $t \in [0,T]$ is given by 
	\begin{equation}\label{cost}
		\min_{Z \in  \mathscr A} \mathbb E\left[ \int_0^T\left( Y_{s-}\,dZ_s+\frac{\gamma_{2}}{2}\,d[Z]_s + \sigma_s d[Z,W]_s \right)+\int_0^T\lambda X^2_s\,ds\right]
	\end{equation}
	subject to the state dynamics
	\begin{equation}\label{model}
		\left\{\begin{split}
			&~dX_s=-dZ_s\\
			&~dY_s=-\rho Y_s\,ds+\gamma_{2}\,dZ_s+\sigma_s\,dW_s\\
			&~X_{0-} = x_0; ~ Y_{0-}=0; ~ X_T = 0
		\end{split}\right.
	\end{equation}
	for $0 \leq s \leq T$, where the set of admissible controls is given by the set 
	\begin{equation*}
		\mathscr A:=\left\{Z: Z\textrm{ is a c\`adl\`ag } \mathbb F\textrm{ semimartingale with } {\mathbb E[[Z]_T]}<\infty	\textrm{ and }{\mathbb E\left[\sup_{0\leq t\leq T}Z^2_t\right]<\infty}	\right\}
	\end{equation*}
	of all square integrable semi-martingales with finite expected covariation. 
	
	The process $(X_t)_{t \in [0,T]}$ describes the investor's portfolio holdings (``inventory''){ with the terminal state constraint $X_T=0$ reflecting the liquidation constraint.} The process $(Z_t)_{t \in [0,T]}$ is the trading strategy
	and $(Y_t)_{t \in [0,T]}$  specifies the transient impact{ of the investor's past trading on future transaction prices.} We think of $Y$ as an additional drift ``added'' to an unaffected martingale price process $M$ that increases linearly in the investors order flow, changes continuously due to noise trader effects, and recovers from past trading at a constant rate $\rho \geq 0$. That is, we assume that the transaction price process is of the form 
	\begin{equation}\label{price}
		S_t = M_t - Y_t, \quad t \in [0,T].
	\end{equation}  
	The martingale part drops out of the equation when computing the expected cost from trading, which are hence given by $\mathbb E\left[ \int_0^T Y_{s-}\,dZ_s \right]$. This explains the first term in our cost function. The fourth term captures market risk and penalizes slow liquidation. These two terms are standard in the liquidation literature; see \cite{Alfonsi-2016, AJK-2014, Cartea-2018, FGHP-2018, HN-2014, HXZ-2020, PZ-2018} and references therein.
	
	When working with semimartingale strategies one additionally needs to penalize the (co)variation of the semimartinagle strategy (with the driving Brownian motion); otherwise the optimization problem is not well posed.{ This observation goes back at least to G\^arleanu and Pedersen \cite{GP}; see also \cite{Ackermann2021,Hkiv,lorenz_schied}.} This explains the second and third term in our cost function. 
	
	
	
	\subsection{Expectations feedback.}
	
	We are now going to introduce an additional feedback effect into the above model that accounts for the possibility of an additional order flow (“child orders”) triggered by the large investor’s trading activities. Specifically, we assume that the transaction price impact process satisfies a dynamics of the form  
	\begin{equation}\label{price2}
		dY_t = -\rho Y_t dt + \gamma_1  dC_t + \gamma_2  dZ_t  + \sigma_t dW_t, \quad t \in [0,T]
	\end{equation}  
	where the process $(C_t)_{t \in [0,T]}$ specifies the large investor's assessment of the impact of child orders triggered by his own trading on the price impact process.  
	
	To motivate the dynamics of the {\it endogenous} child order term, let us assume as in \cite{CHT,FHX1} that market sell and buy orders arrive according to independent exponential Hawkes processes $N^\pm$ with respective intensities
	\[
	\zeta^\pm_t := \mu^\pm_t + \alpha \int_0^t e^{-\beta(t-s)}dN^\pm_s. 
	\]
	Here, $\mu^\pm$ are {\it endogenous} base intensities that will be specified below, and $\alpha,\beta$ are {\it exogenous} deterministic coefficients that capture the impact of past orders on future order flow. The total order flow imbalance can be written as a sum 
	\[
	N^+_t - N^-_t =:  Z_t + {\cal C}_t
	\]
	of the trader's contribution $Z$ and the endogenous child order process $ \cal C$. Working with the child order process is inconvenient. It is much more convenient to consider instead  
	%
	%
	the {\it expected} child order process as a measure for the impact of own trading on future transaction price dynamics in \eqref{price2}, i.e.~to set\footnote{Working with the expected child order process will allow us to work with the child order arrival rate and also accounts for the fact that child orders are not directly observable. We further comment on the choice of the expected child order process below. }  
	\[
	C_t := \mathbb E[{\cal C}_t].
	\]
	
	We assume that the market is in equilibrium - i.e.~that $\mu^\pm \equiv \mu$ - if the large investor is inactive. In this case, the order flow imbalance $N^+ - N^-$ is a martingale and hence $C_t \equiv 0$. If the large investor employs a {\it continuous} semi-martingale strategy of the form
	\[
	dZ_t = b_t dt + \tilde \sigma_t dW_t,
	\]
	it seems natural to change the base intensities to 
	\[
	\mu^\pm_t = \mu +b^\pm_t,
	\]	
	where $b^\pm_t$ denotes the positive/negative part of $b_t$. In other words, we assume that the market order arrival intensities increase in the large investor's trading with the martingale part of his trading strategy having no impact on future order flow. 
	
	Since the compensated Hawkes processes $\left(N^\pm_t - \int_0^t \zeta^\pm_s ds \right)_{0 \leq t \leq T}$ are martingales starting in zero, the expected child order flow equals
	\begin{equation}
		C_t = \mathbb E \left[ \int_0^t \zeta^+_s ds - \int_0^t \zeta^-_s ds - Z_t \right].
	\end{equation}
	With our choice of base intensities the expected intensities satisfy
	\[
	\mathbb E [\zeta^\pm_t] = \mu + \mathbb E [b^\pm_t] + \alpha \int_0^t e^{-\beta(t-s)}  \mathbb E [\zeta^\pm_s] ds. 	
	\]
	As a result, the expected aggregate order flow imbalance $N_t :=  \mathbb E[N^+_t - N^-_t] $ equals  
	\begin{equation}
		\begin{split}
			N_t 
			& = ~ \int_0^t \mathbb E[b_s]ds + \alpha \int_0^t e^{-\beta(t-s)} N_s ds \\
			& = ~ \mathbb E[Z_t]  + \alpha \int_0^t e^{-\beta(t-s)} N_s ds,
		\end{split}
	\end{equation}
	and hence the expected number of child orders equals
	\begin{equation*}
		\begin{split}
			C_t 
			& = ~ \alpha \int_0^t e^{-\beta(t-s)} N_s ds.  
		\end{split}
	\end{equation*}
	%
	
	Differentiating the above equation shows that\footnote{Essentially, it makes no difference to consider $x_0$ as random or deterministic. Here, we assume it is a square integrable r.v.. }  
	\begin{equation}\label{dC}
		dC_t = \left( -(\beta-\alpha)C_t + \alpha (\mathbb E[x_0] -  \mathbb E [X_t] ) \right) dt, \quad C_0=0. 
	\end{equation}
	In particular, the child order flow rate resulting from the order submission rate $b$ increases linearly in the expected traded volume $\mathbb E[x_0] -  \mathbb E[X_t]$ and is mean-reverting if $\frac{\alpha}{\beta} < 1$.  In this case, each order triggers less than one child order on average and the Hawkes process is stable in the long run; see \cite{HO} for details.   
	
	Choosing the child order dynamics \eqref{dC} in \eqref{price2}  our transient impact process follows the dynamics
	\begin{equation}\label{model-t1}
		\left\{\begin{split}
			&~dY_t=\Big(-\rho Y_t+\gamma_{1} \dot C_t\Big)\,dt+\gamma_{2}\,dZ_t + \sigma_t\,dW_t\\
			&~dC_t=-(\beta-\alpha)C_t dt+\alpha(\mathbb E[x_0]-\mathbb E[X_t]) \,dt\\
			&~Y_{0}=0; \quad C_0 = 0. 
		\end{split}\right.
	\end{equation}
	
	In particular, the dynamics of the transient impact process $Y$ can be decomposed into three components\footnote{Note that a similar decomposition also appears in \cite{MuhleKarbe}.}, namely, the transient impact of the larger trader's own trading activity
	\[
	dY^1_t = -\rho Y^1_t\,dt+\gamma_2\,dZ_t,
	\]
	the transient impact of the child orders
	\[
	dY^2_t= -\rho Y^2_t\,dt+\gamma_1\,dC_t,
	\]
	and the transient impact of noise traders' activities
	\[
	dY^3_t = -\rho Y^3_t\,dt + \sigma_t\,dW_t.
	\]

	\begin{remark}
At first sight one may be inclined to work with the child order process $\cal C$ in \eqref{price2} and hence in \eqref{model-t1}. We prefer to work instead with the expected order flow $C$ for at least two reasons. First, working with the order flow renders our model analytically tractable. Second, the Hawkes processes $N^\pm$ and hence the child order processes are unobservable in general. While the trader cannot observe the child order process he/she can estimate the base intensity $\mu$ along with the parameters $\alpha$ and $\beta$ from past data and hence estimate the intensities $\zeta^\pm$. This allows the trader to compute the expected intensities and thus to specify the expected order flow, which strikes us as a canonical way to measure own impact. We acknowledge that it would be more even canonical to work with the \textsl{conditional} expected order flow given the trader's own submission rate $b$. However, to the best of our knowledge there is no closed form expression for the conditional expected flow.    
	\end{remark}
}

%
%
%
%
%


\subsection{The optimization problem}

We are now ready to state the large trader's optimization problem when general semi-martingale strategies are allowed. We retain the assumption that the cost functional at time $t \in [0,T)$ is given by
\[
	J(t,Z) = \mathbb E \left[  \int_t^T\left( Y_{s-}\,dZ_s+\frac{\gamma_{2}}{2}\,d[Z]_s + \sigma_s d[Z,W]_s \right)+\int_t^T\lambda X^2_s\,ds \right].
\] 
{We also assume that the child order flow is mean-reverting and that it increases linearly in the expected traded volume $\mathbb E[x_0] -  \bE[X_t]$. In other words, we assume that the impact process satisfies \eqref{price2}. The trader's control problem at time $t \in [0,T]$ is thus given by  }
\begin{equation}\label{cost-t}
	\min_{Z \in  \mathscr A_t}J(t,Z)
\end{equation}
subject to the state dynamics
\begin{equation}\label{model-t}
	\left\{\begin{split}
		&~dX_s=-dZ_s\\
		&~dY_s=\Big(-\rho Y_s+\gamma_{1} C'_s\Big)\,ds+\gamma_{2}\,dZ_s+\sigma_s\,dW_s\\
		&~dC_s=-(\beta-\alpha)C_s ds+\alpha(\mathbb E[x_0]-\mathbb E[X_s]) \,ds\\
		&~(X_{t-}, Y_{t-},C_{t-})=\mathcal X; ~ X_T = 0,
	\end{split}\right.
\end{equation}
where the set of admissible trading strategies is given by. 
\begin{equation}
\begin{split}
	\mathscr A_t & := ~ \bigg\{Z: Z\textrm{ is a c\`adl\`ag } \mathbb F\textrm{ semimartingale starting at $t$ with }  {\mathbb E\left[[Z]_T\right]<\infty} \textrm{ and }   \\
	& \left. \qquad \qquad  \mathbb E\left[\sup_{t\leq s\leq T}|Z_s|^2\right]<\infty 	\right\}.
\end{split}
\end{equation}

{
\begin{remark}
\begin{itemize}
	\item[i)] 
	We assume that large block trades increase the rate of child order arrivals but do not trigger ``block child orders''. Block trades following block trades would be executed at very unfavorable prices. Hence, it is natural to assume that other traders would not submit block trades in response to observed block trades but much rather increase their order submission rates. 
	\item[(ii)] Cay\'e and Muhle-Karbe \cite{Caye-2016} consider a model with self-exciting order flow where ``other (e.g., predatory) traders augment the order flow of the large trader by a linear factor depending on $x_0 - X_t$''. Since theirs is a deterministic model there is no difference between traded and expected traded volume. By contrast, we allow for stochastic trading strategies.  For the reasons outlined above this suggests to use expected rather than actual trading volume as a measure for own impact on child order flow.   
\end{itemize}
\end{remark}
}


\subsubsection{The value function}

In what follows we denote the value function of our control problem, given the random  initial state ${\cal X} = (X_{t-},Y_{t-},C_{t-})$ at time $t \in [0,T)$, by 
\begin{equation}\label{def:value-function-t}
	V(t,\mathcal X) = \inf_{Z \in {\mathscr A}_t }J(t,Z). 	
\end{equation}

The challenge is to identify a candidate for the value function. Although it is natural to expect the value function to be of linear-quadratic form, the dynamics of the respective coefficient processes is far from obvious.   

Our goal is to represent the value function in terms of two deterministic symmetric $\mathbb R^{3\times 3}$-valued processes $A, B$, an $\mathbb R^3$-valued deterministic process $D$ and a real-valued adapted square integrable process $F$ as\footnote{The notation $\text{Var}({\mu})(A_{t})$ was introduced at the end of the introduction.} 
\begin{equation}\label{value1}
	V(t,\mathcal X)=\text{Var}({\mu})(A_{t}) + \bar{\mu}^\top B_t\bar{\mu}+D^\top_t\bar{\mu}+\mathbb E[F_t]
\end{equation}
where $\mu$ denotes the law of the initial state $\cal X$, and $\bar \mu$ denotes the vector of expected values. 

The dynamics of the processes $A, B, D, F$ is derived heuristically in Appendix A by first analyzing a discrete time model and then taking the limit as the time difference between two consecutive trading periods tends to zero. It turns out that: 

\begin{itemize}
\item The process $A$ is symmetric, satisfies $A_{11}=\gamma_{2}A_{21}$, $A_{12}=\gamma_{2}A_{22}+\frac{1}{2}$, $A_{13}=\gamma_{2}A_{23}$, and  the ODE
\begin{equation}\label{eq:tilde-A}
	\left\{
	\begin{aligned}
		\dot A_{11,t}=&\left(-\lambda+\frac{\left(\rho A_{11,t}+\lambda\right)^2}{\gamma_2\rho+\lambda}\right)\\
		\dot A_{13,t}=&\bigg(\frac{\gamma_1(\beta-\alpha)}{\gamma_2}A_{11,t}+(\beta-\alpha)A_{13,t}-\frac{\left(\rho A_{11,t}+\lambda\right)\left(\gamma_1(\beta-\alpha)-2\rho A_{13,t}\right)}{2(\gamma_2\rho+\lambda)}\bigg)\\
		\dot A_{33,t}=&\left(2(\beta-\alpha)A_{33,t}+2\frac{\gamma_1(\beta-\alpha)}{\gamma_2}A_{13,t}+\frac{\left(\gamma_1(\beta-\alpha)-2\rho A_{13,t}\right)^2}{4(\gamma_2\rho+\lambda)}\right)\\
		A_{11,T}=&\frac{\gamma_2}{2},\quad A_{13,T}=0,\quad A_{33,T}=0.
	\end{aligned}\right.
\end{equation}
The above is a standard ODE system that can be uniquely solved. 

	\item The process $B$ is symmetric, satisfies $B_{11}=\gamma_{2}B_{21}$, $B_{12}=\gamma_{2}B_{22}+\frac{1}{2}$, $B_{13}=\gamma_{2}B_{23}$, and the fully coupled system of Riccati-type equations
\begin{equation}\label{eq:tilde-B}
	\left\{
	\begin{aligned}
		\dot B_{11,t}=&~\Bigg(2\frac{\gamma_1\alpha}{\gamma_2}B_{11,t}+2\alpha B_{13,t}-\lambda   \\
		&   +\frac{\Big(2(\gamma_1\alpha-\gamma_2\rho) B_{11,t}+\gamma_1\gamma_2\alpha+2\alpha\gamma_2B_{13,t}-2\gamma_2\lambda\Big)^2}{4\gamma_2^2(\gamma_2\rho-\gamma_1\alpha+\lambda)}\Bigg)\\
		\dot B_{33,t}=&~\Bigg(2(\beta-\alpha)B_{33,t}+2\frac{\gamma_1(\beta-\alpha)}{\gamma_2}B_{13,t}   \\ 
		&  +\frac{\Big(2(\gamma_1\alpha-\gamma_2\rho)B_{13,t}+2\gamma_2\alpha B_{33,t}+\gamma_1\gamma_2(\beta-\alpha)\Big)^2}{4\gamma_2^2(\gamma_2\rho-\gamma_1\alpha+\lambda)}\Bigg)\\
		\dot B_{13,t}=&~\Bigg\{\frac{\gamma_1(\beta-\alpha)}{\gamma_2}B_{11,t}+\alpha B_{33,t}+(\beta-\alpha+\frac{\gamma_1\alpha}{\gamma_2})B_{13,t}\\
		&+\Big(2(\gamma_1\alpha-\gamma_2\rho) B_{11,t}+\gamma_1\gamma_2\alpha+2\alpha\gamma_2B_{13,t}-2\gamma_2\lambda\Big)\\
		&\cdot\frac{\left(2(\gamma_1\alpha-\gamma_2\rho)B_{13,t}+2\gamma_2\alpha B_{33,t}+\gamma_1\gamma_2(\beta-\alpha)\right)}{4\gamma_2^2(\gamma_2\rho-\gamma_1\alpha+\lambda)}\Bigg\}\\
		 B_{11,T}=&~\frac{\gamma_2}{2},\quad B_{13,T}=0,\quad B_{33,T}=0.
	\end{aligned}\right.
\end{equation}
This system is complicated to analyze; its analysis is postponed to the next section. 

\item The vector-valued process $D$ satisfies $D_{2} = \gamma_2^{-1} D_{1}$, and the components $D_1$ and $D_3$ satisfy the coupled linear ODE system 
\begin{equation}\label{eq:tilde-D}
	\left\{\begin{split}
		\dot D_{1,t}=&~\Bigg\{-\frac{2\gamma_1\alpha \mathbb E[x_0]}{\gamma_2}B_{11,t}-2\alpha \mathbb E[x_0] B_{13,t}+\frac{\gamma_1\alpha}{\gamma_2} D_{1,t}+\alpha D_{3,t}\\
		&+\Big(-2\lambda\gamma_2+2(\gamma_1\alpha-\gamma_2\rho)B_{11,t}+\gamma_1\gamma_2\alpha+\alpha\gamma_2B_{13,t}\Big)\\
		&\cdot\frac{\Big(-\gamma_1\gamma_2\alpha \mathbb E[x_0]+(\gamma_1\alpha-\gamma_2\rho)D_{1,t}+\alpha\gamma_2 D_{3,t}\Big)}{2\gamma_2^2(\gamma_2\rho-\gamma_1\alpha+\lambda)}\Bigg\}\\
		\dot D_{3,t}=&~\Bigg\{-2\alpha \mathbb E[x_0] B_{33,t}-\frac{2\gamma_1\alpha \mathbb E[x_0]}{\gamma_2}B_{13,t}+\frac{\gamma_1(\beta-\alpha)}{\gamma_2}D_{1,t}+(\beta-\alpha)D_{3,t}\\
		&+\Big(2(\gamma_1\alpha-\gamma_2\rho)B_{13,t}+2\gamma_2\alpha B_{33,t}+\gamma_1\gamma_2(\beta-\alpha)\Big)\\
		&\cdot\frac{\Big(-\gamma_1\gamma_2\alpha \mathbb E[x_0]+(\gamma_1\alpha-\gamma_2\rho)D_{1,t}+\alpha\gamma_2 D_{3,t}\Big)}{2\gamma_2^2(\gamma_2\rho-\gamma_1\alpha+\lambda)}\Bigg\}\\
		D_{1,T}=&~D_{3,T}=0
	\end{split}\right.
\end{equation}

\item Due to the random volatility process, the process $F$ satisfies a BSDE, namely 
	\begin{equation}\label{eqn:F}
	\left\{\begin{split}
		-d F_t=&~\bigg\{\sigma^2_t \frac{2A_{11,t}-\gamma_2}{2\gamma_2^2}+\alpha\gamma_1\mathbb E[x_0]\frac{D_{1,t}}{\gamma_2}+\alpha\mathbb E[x_0]D_{3,t}\\
	&-\frac{1}{4(\lambda+\gamma_2\rho-\alpha\gamma_1)}\left(-\alpha\gamma_1\mathbb E[x_0]+(\gamma_1\alpha-\gamma_2\rho)\frac{D_{1,t}}{\gamma_2}+\alpha D_{3,t}\right)^2\bigg\}\,dt-Z^F_t\,dW_t\\
	F_T=&~0.
	\end{split}\right.
\end{equation}
\end{itemize}

We prove in Section 3 that the system \eqref{eq:tilde-B} is well posed and admits a unique global solution if the feedback effect as measured by the constant $\alpha$ is weak enough. In this case, the joint dynamics of the process $(A,B,D,F)$ is well-defined. 


\subsubsection{The main result}

We assume throughout that the following {\bf standing assumption} holds.
 
\begin{enumerate}
	\item The coefficients $\gamma_1$, $\gamma_2$, $\alpha$, $\beta$, $\rho$ and $\lambda$ are nonnegative constants. 
	\item The coefficients satisfy $\beta-\alpha>0$ and $\gamma_2\rho-\gamma_1\alpha+\lambda>0$. 
	\item The initial position $x_0$ is a square integrable r.v.~that is independent of the Brownian motion. The volatility process $\sigma$ is a square integrable progressively measurable process.\footnote{We emphasize that $\sigma$ is allowed to be degenerate.}
\end{enumerate}

It will be convenient to rewrite the state dynamics and the cost function in matrix form as
\begin{equation}\label{state_matrix}
\begin{split}
	d\mathcal X_s &=\Big(	\mathcal H\mathcal X_s+\overline{\mathcal H}\mathbb E[\mathcal X_s]	+\mathcal G	\Big)\,ds+\mathcal D_s\,dW_s+\mathcal K\,dZ_s,\quad s\in[t,T), \\
	\mathcal X_{t-} & = \mathcal X
\end{split}
\end{equation}
where 
\begin{equation}
	\begin{split}
		\mathcal H=&~ \begin{pmatrix}		0&0&0\\
			0& -\rho&-\gamma_{1}(\beta-\alpha)\\
			0& 0& -(\beta-\alpha)
		\end{pmatrix},\quad ~~  \overline{\mathcal H}=  \begin{pmatrix}			
			0&0&0\\
			-\alpha\gamma_{1}&0&0\\
			-\alpha&0&0
		\end{pmatrix},\\
		\mathcal G=&~\begin{pmatrix}		 0&\alpha\gamma_{1}\mathbb E[x_0]& \alpha\mathbb E[x_0]		\end{pmatrix}^\top,\quad  \mathcal D_s=\begin{pmatrix}	0&\sigma_s&0				\end{pmatrix}^\top,\\
		\mathcal K=&~\begin{pmatrix}
			-1 	& \gamma_{2}& 0	\end{pmatrix}^\top.\\
	\end{split}
\end{equation}

The following is the main result of this paper. Its proof is given in Sections 3 and 4 below. 

\begin{theorem}
If the standing assumption is satisfied, and if either $\lambda=0$ or $\lambda\rho\gamma_2 > 0$ and $\alpha$ is small enough, then the following holds.  
\begin{itemize}
\item[i)] In terms of the processes $A,B,D,F$ introduced in \eqref{eq:tilde-A}-\eqref{eqn:F}  the value function defined in \eqref{def:value-function-t} is given by 
\begin{equation} \label{value2}
	V(t,\mathcal X)=V(t,\mu)=\text{Var}({\mu})(A_{t}) + \bar{\mu}^\top B_t\bar{\mu}+D^\top_t\bar{\mu}+\mathbb E[F_t]. 
\end{equation}

\item[ii)] The optimal strategy $\widetilde Z$ jumps only at the beginning and the end of the trading period where the initial and terminal jump is given by 
\begin{equation} \label{opt1}
	\Delta \widetilde Z_t = -\frac{I_{t}^A}{\tilde a}(\mathcal X_{t-}-\overline\mu)-\frac{I_{t}^B}{a}\overline\mu-\frac{I_{t}^D}{a} 
	\quad \mbox{and} \quad 
	\Delta \widetilde Z_T = X_{T-} 
\end{equation}	
respectively. On the time interval $(t,T)$ the optimal strategy satisfies the dynamics 
\begin{equation} \label{opt2}
	\begin{split}
	d \widetilde Z_s=&\left(-\frac{\dot I^A_s}{\tilde a}({\mathcal X_s}-\mathbb E[{\mathcal X_s}])-\frac{\dot I^B_s}{a}\mathbb E[{\mathcal X_s}]-\frac{\dot I^D_s}{a}-\frac{I^A_s}{\tilde a}\mathcal H({\mathcal X_s}-\mathbb E[{\mathcal X_s}]) \right. \\ 
	& - \left. \frac{I^B_s}{a}((\mathcal H+\overline{\mathcal H})\mathbb E[{\mathcal X}_s]+\mathcal G)\right)\,ds
-\frac{I^A_s}{\tilde a}\mathcal D_s\,dW_s,\qquad \qquad s\in(t,T)
	\end{split}
\end{equation}
where $\tilde a=\gamma_2\rho+\lambda$, $a=\gamma_2\rho-\gamma_1\alpha+\lambda$, and the processes $I^A,I^B$ and $I^D$ are given by 
\begin{equation*}
		I^A=\begin{pmatrix}
			-\rho A_{11}-\lambda\\
			-\rho \frac{A_{11}}{\gamma_2}+\rho\\
			\frac{\gamma_1(\beta-\alpha)}{2}-\rho A_{13}
		\end{pmatrix}^\top, \quad
		I^B=\begin{pmatrix}
			\frac{\alpha\gamma_1-\gamma_2\rho}{\gamma_2}B_{11}+\alpha B_{13}-\lambda+\frac{\alpha\gamma_1}{2}\\
			\frac{\alpha\gamma_1-\gamma_2\rho}{\gamma_2^2}B_{11}+\alpha \frac{B_{13}}{\gamma_2}+\rho-\frac{\alpha\gamma_1}{2\gamma_2}\\
			\frac{\gamma_1(\beta-\alpha)}{2}+(\gamma_1\alpha-\gamma_2\rho)\frac{B_{13}}{\gamma_2}+\alpha B_{33}
		\end{pmatrix}^\top
\end{equation*}
and
\[
		I^D=	-\frac{\alpha\gamma_1}{2}\mathbb E[x_0]+(\gamma_1\alpha-\gamma_2\rho)\frac{D_1}{2\gamma_2}+\frac{\alpha}{2} D_3.
\]
\end{itemize}
\end{theorem}



\section{Wellposedness of the Riccati Equation }\label{sec-ric}

In this section, we prove that the system \eqref{eq:tilde-B}  is well posed and has a unique global solution. 
Specifically, we prove the following result.
\begin{theorem}\label{thm:risk-averse}
	In addition to the {\bf standing assumption}, let us assume that $\alpha$ is small enough and that $\lambda,\gamma_2,\rho>0$. Then the matrix Riccati equation \eqref{eq:tilde-B} admits a unique solution 
\[	
	B \in L^\infty([0,T];\mathbb R^3)\cap C([0,T];\mathbb R^3). 
\]
\end{theorem}

To prove Theorem \ref{thm:risk-averse}, it will be convenient to introduce the matrix-valued processes
\begin{equation*}
	\begin{aligned}
		\mathcal P&=\left(\begin{matrix}
			B_{11}&B_{13}\\
			B_{13}&B_{33}
		\end{matrix}\right),\quad	
		\mathcal N_1 =\left(\begin{matrix}
			\frac{\gamma_1\alpha}{\gamma_2}+\frac{(\gamma_1\alpha-\gamma_2\rho)(\gamma_1\alpha-2\lambda)}{2\gamma_2(\gamma_2\rho-\gamma_1\alpha+\lambda)}&\alpha+\frac{\alpha(\gamma_1\alpha-2\lambda)}{2(\gamma_2\rho-\gamma_1\alpha+\lambda)}\\
			\frac{\gamma_1(\beta-\alpha)}{\gamma_2}+\frac{(\gamma_1\alpha-\gamma_2\rho)\gamma_1(\beta-\alpha)}{2\gamma_2(\gamma_2\rho-\gamma_1\alpha+\lambda)}&\beta-\alpha+\frac{\gamma_1\alpha(\beta-\alpha)}{2(\gamma_2\rho-\gamma_1\alpha+\lambda)}
		\end{matrix}\right),\quad
		\mathcal N_2=\begin{pmatrix}
			\gamma_1\alpha-\gamma_2\rho\\
			\gamma_2\alpha
		\end{pmatrix},\\
		\mathcal N_0&=\frac{1}{\gamma_2^2(\gamma_2\rho-\gamma_1\alpha+\lambda)}, \qquad \mathcal M=-\begin{pmatrix}
			\frac{\gamma_1^2\alpha^2-4\lambda\gamma_2\rho}{4(\gamma_2\rho-\gamma_1\alpha+\lambda)}&\frac{\gamma_1(\beta-\alpha)(\gamma_1\alpha-2\lambda)}{4(\gamma_2\rho-\gamma_1\alpha+\lambda)}\\
			\frac{\gamma_1(\beta-\alpha)(\gamma_1\alpha-2\lambda)}{4(\gamma_2\rho-\gamma_1\alpha+\lambda)}&\frac{\gamma_1^2(\beta-\alpha)^2}{4(\gamma_2\rho-\gamma_1\alpha+\lambda)},
		\end{pmatrix},
		\qquad \mathcal G=\begin{pmatrix}
			\gamma_2 & 0\\
			0 & 0
		\end{pmatrix},
	\end{aligned}
\end{equation*}
so that the system \eqref{eq:tilde-B} can be rewritten in the matrix form as:
\begin{equation}\label{eq:multi-Riccati}
	\left\{
	\begin{aligned}
		\dot{\mathcal P}_t&=\left(\mathcal P_t\mathcal N_2\mathcal N_0\mathcal N_2^\top\mathcal P_t+\mathcal N_1\mathcal P_t+\mathcal P_t\mathcal N_1^\top -\mathcal M\right),\quad t\in[0,T)\\
		\mathcal P_T&=\mathcal G.
	\end{aligned}\right.
\end{equation}
The matrix-valued Riccati equation \eqref{eq:multi-Riccati} does not satisfy the requirements of \cite[Proposition 2.1, 2.2]{Kohlmann2003} as $\cal M$ is not positive semi-definite. To overcome this problem we employ a sophisticated transformation to bring the equation into standard Riccati-form. To this end, we define
\[
			\widetilde{\mathcal P}=\mathcal P+\widetilde\Lambda,
\]   
	where the matrix $\widetilde \Lambda$ is given by 
	\begin{equation*}
		\widetilde{\Lambda}=
		\begin{pmatrix}	\lambda_1&\lambda_2\\
			\lambda_2& \lambda_3			\end{pmatrix}=
		\begin{pmatrix}	\Lambda\alpha^2&\Lambda\alpha(\beta-\alpha)\\
			\Lambda\alpha(\beta-\alpha)& \Lambda(\beta-\alpha)^2		\end{pmatrix}
	\end{equation*}
	for some constant $\Lambda$ that will be determined in what follows. The process $\widetilde{\mathcal P}$ satisfies the dynamics 
	\begin{equation}\label{eq:Riccati-tilde-app}
		\begin{split}
			\dot{\widetilde{\mathcal P}}_t = & ~ \widetilde{\mathcal P}_t\mathcal N_2\mathcal N_0\mathcal N_2^\top\widetilde{\mathcal P}_t + \widetilde{\mathcal N_1}\widetilde{\mathcal P}_t+\widetilde{\mathcal P}_t\widetilde{\mathcal N_1}^\top -\widetilde{\mathcal M},\quad t\in[0,T) \\
			\widetilde{\mathcal P}_T = & ~ \mathcal G+\widetilde{\Lambda},
		\end{split}
	\end{equation}
	where the matrices $\widetilde{\mathcal N_1}$ and $\widetilde{\mathcal M}$ are given by, respectively, 
	\begin{equation*}
		\widetilde{\mathcal N_1}=\mathcal N_1-	\widetilde{\Lambda}\mathcal N_2\mathcal N_0\mathcal N_2^\top \qquad \mbox{and} \qquad
		\widetilde{\mathcal M}=-\widetilde{\Lambda}\mathcal N_2\mathcal N_0\mathcal N_2^\top\widetilde{\Lambda} +\mathcal N_1\widetilde{\Lambda} +\widetilde{\Lambda}\mathcal N_1^\top+\mathcal M.
	\end{equation*}
It is enough to prove that the above matrix-valued ODE has a unique solution for a suitable $\Lambda \in \mathbb R$. 

\vspace{2mm}

\textsc{Proof of Theorem \ref{thm:risk-averse}}.	We are going to show that the equation \eqref{eq:Riccati-tilde-app} satisfies the assumptions of \cite[Proposition 2.1, 2.2]{Kohlmann2003}, that is that the matrix $\widetilde{\cal M}$ is positive semidefinite for a suitably chosen constant $\Lambda$. The entries of $\widetilde{\cal M}$ are given by, respectively, 
		\begin{align*}
			\widetilde{\mathcal M}_{11}=&
			-\frac{	 \Big\{	\lambda_1(\gamma_1\alpha-\gamma_2\rho)+\lambda_2\gamma_2\alpha		\Big\}^2	}{\gamma_2^2(\gamma_2\rho-\gamma_1\alpha+\lambda)} + 2\lambda_1\left(\frac{\gamma_1\alpha}{\gamma_2}+\frac{(\gamma_1\alpha-\gamma_2\rho)(\gamma_1\alpha-2\lambda)}{2\gamma_2(\gamma_2\rho-\gamma_1\alpha+\lambda)} \right) \\
			&+2\lambda_2\left(		\alpha+\frac{\alpha(\gamma_1\alpha-2\lambda)}{2(\gamma_2\rho-\gamma_1\alpha+\lambda)}				\right) -\frac{\gamma_1^2\alpha^2-4\lambda\gamma_2\rho}{4(\gamma_2\rho-\gamma_1\alpha+\lambda)},     \\
			\widetilde{\mathcal M}_{12} = & ~ \widetilde {\mathcal M}_{21}\\
			=&~ -\frac{	\Big\{	 \lambda_1(	\gamma_1\alpha-\gamma_2\rho	)+\lambda_2\gamma_2\alpha  				\Big\} \Big\{		\lambda_2(	\gamma_1\alpha-\gamma_2\rho		)+\lambda_3\gamma_2\alpha			\Big\}		}{ \gamma_2^2(\gamma_2\rho-\gamma_1\alpha+\lambda) }\\ &~+\lambda_2\left(\frac{\gamma_1\alpha}{\gamma_2}+\frac{(\gamma_1\alpha-\gamma_2\rho)(\gamma_1\alpha-2\lambda)}{2\gamma_2(\gamma_2\rho-\gamma_1\alpha+\lambda)} \right) 
			+\lambda_3\left(		\alpha+\frac{\alpha(\gamma_1\alpha-2\lambda)}{2(\gamma_2\rho-\gamma_1\alpha+\lambda)}				\right) \\
			&~+ \lambda_1\left(\frac{\gamma_1(\beta-\alpha)}{\gamma_2}+\frac{(\gamma_1\alpha-\gamma_2\rho)\gamma_1(\beta-\alpha)}{2\gamma_2(\gamma_2\rho-\gamma_1\alpha+\lambda)} \right) +\lambda_2\left(	\beta-\alpha+\frac{\gamma_1\alpha(\beta-\alpha)}{2(\gamma_2\rho-\gamma_1\alpha+\lambda)}					\right)\\
			&~-           \frac{\gamma_1(\beta-\alpha)(\gamma_1\alpha-2\lambda)}{4(\gamma_2\rho-\gamma_1\alpha+\lambda)},\\
			\widetilde{\mathcal M}_{22}= &~-\frac{ \Big\{\lambda_2(\gamma_1\alpha-\gamma_2\rho)+\lambda_3\gamma_2\alpha		\Big\}^2 }{\gamma_2^2(\gamma_2\rho-\gamma_1\alpha+\lambda)}+2\lambda_2\left(\frac{\gamma_1(\beta-\alpha)}{\gamma_2}+\frac{(\gamma_1\alpha-\gamma_2\rho)\gamma_1(\beta-\alpha)}{2\gamma_2(\gamma_2\rho-\gamma_1\alpha+\lambda)} \right)  \\
			&~  +2\lambda_3\left(	\beta-\alpha+\frac{	\gamma_1\alpha(\beta-\alpha)	}{2(\gamma_2\rho-\gamma_1\alpha+\lambda)}			\right)- \frac{	\gamma_1^2(\beta-\alpha)^2	}{ 4(\gamma_2\rho-\gamma_1\alpha+\lambda)			}.
	\end{align*}
	
	In terms of the functions 
		\begin{equation*}
		\begin{split}
			f(\Lambda) := &~-\frac{	\alpha^2\Big(	\gamma_1\alpha-\gamma_2\rho+\gamma_2(\beta-\alpha)	\Big)^2	}{\gamma_2^2(\gamma_2\rho-\gamma_1\alpha+\lambda)} \Lambda^2   + 2\left(\frac{\gamma_1\alpha}{\gamma_2}+\frac{\gamma_1\alpha(\gamma_1\alpha-\gamma_2\rho)}{2\gamma_2(\gamma_2\rho-\gamma_1\alpha+\lambda)} +\beta-\alpha \right. \\
			& ~~ \left. +\frac{\gamma_1\alpha(\beta-\alpha)}{2(\gamma_2\rho-\gamma_1\alpha+\lambda)}				\right)\Lambda
			+ \frac{\gamma_1^2}{4(\gamma_2\rho-\gamma_1\alpha+\lambda)}, \\
			g(\Lambda) :=&~\lambda\left(	\frac{\gamma_1\alpha-\gamma_2\rho}{\gamma_2(\gamma_2\rho-\gamma_1\alpha+\lambda)}+\frac{\beta-\alpha}{\gamma_2\rho-\gamma_1\alpha+\lambda}\right)\Lambda,
		\end{split}
	\end{equation*}
the matrix $\widetilde{\cal M}$ can be written as
%
%
%
	\begin{equation*}
		\begin{pmatrix}
			 f(\Lambda)\alpha^2-2\alpha^2g(\Lambda)+\frac{\lambda\gamma_2\rho}{\gamma_2\rho-\gamma_1\alpha+\lambda}&f(\Lambda)\alpha(\beta-\alpha)-\alpha(\beta-\alpha) g(\Lambda)+\frac{\lambda\gamma_1(\beta-\alpha)}{2(\gamma_2\rho-\gamma_1\alpha+\lambda)}\\
			f(\Lambda)\alpha(\beta-\alpha)-\alpha(\beta-\alpha) g(\Lambda)+\frac{\lambda\gamma_1(\beta-\alpha)}{2(\gamma_2\rho-\gamma_1\alpha+\lambda)}&f(\Lambda)(\beta-\alpha)^2
		\end{pmatrix}.
	\end{equation*}

	Straightforward calculations show that 
	\begin{equation*}
		\begin{split}
			& ~ \textrm{Det}[\widetilde{\mathcal M}] \\
			=&~(\beta-\alpha)^2\left(f(\Lambda)\frac{\lambda\gamma_2\rho}{\gamma_2\rho-\gamma_1\alpha+\lambda}-\frac{\lambda^2\gamma_1^2}{4(\gamma_2\rho-\gamma_1\alpha+\lambda)^2}-\alpha^2g^2(\Lambda)-(f(\Lambda)-g(\Lambda))\frac{\lambda\gamma_1\alpha}{\gamma_2\rho-\gamma_1\alpha+\lambda}\right)\\
			=&~\frac{\lambda\gamma_2\rho(\beta-\alpha)^2}{\gamma_2\rho-\gamma_1\alpha+\lambda}\left(f(\Lambda)-\frac{\lambda\gamma_1^2}{4\gamma_2\rho(\gamma_2\rho-\gamma_1\alpha+\lambda)}\right)+O(\alpha).
		\end{split}
	\end{equation*}
	The proof of the positive semi-definiteness of $\widetilde{\cal M}$ is now split into the following two cases.
	
	\begin{itemize}
	\item \textbf{Case 1. $\gamma_1\alpha-\gamma_2\rho+\gamma_2(\beta-\alpha)\leq0.$} In this case, 
	\begin{equation}\label{eq:g<0}
		g(\Lambda)\leq 0\quad \textrm{ for all } \quad \Lambda\geq 0
	\end{equation}
	and we put 
	\[
		h_1(\Lambda):=f(\Lambda)-\frac{\lambda\gamma_1^2}{4\gamma_2\rho(\gamma_2\rho-\gamma_1\alpha+\lambda)}.
	\] 

	\vspace{4mm}	
	
	\textbf{Case 1.1. $\alpha(\gamma_1\alpha-\gamma_2\rho+\gamma_2(\beta-\alpha))=0$.} In this case either $\alpha=0$ or $\gamma_1\alpha-\gamma_2\rho+\gamma_2(\beta-\alpha)=0$ and hence $h_1$ is linear with a positive leading coefficient. Thus, choosing $\Lambda=\Lambda_0>0$ large enough, $h_1(\Lambda_0)>0$ and thus $f(\Lambda_0)>0$ as well. Moreover, 
	\begin{equation*}
		\begin{split}
			\widetilde {\mathcal M}_{11}=f(\Lambda_0)\alpha^2-2g(\Lambda_0)\alpha^2+\frac{\lambda\gamma_2\rho}{\gamma_2\rho-\gamma_1\alpha+\lambda}>0,\quad\widetilde {\cal M}_{22}=f(\Lambda_0)(\beta-\alpha)^2>0\\
		\end{split}
	\end{equation*}
	and
	\begin{equation*}
		\textrm{Det}[\widetilde {\cal M}]=\frac{\lambda\gamma_2\rho(\beta-\alpha)^2}{\gamma_2\rho-\gamma_1\alpha+\lambda}\left(f(\Lambda_0)-\frac{\lambda\gamma_1^2}{4\gamma_2\rho(\gamma_2\rho-\gamma_1\alpha+\lambda)}\right)+O(\alpha)>0
	\end{equation*}
	by choosing $\Lambda_0$ large enough. Hence, in this case our matrix is positive semi-definite.
	
	\vspace{4mm}	
	
	\textbf{Case 1.2. $\alpha(\gamma_1\alpha-\gamma_2\rho+\gamma_2(\beta-\alpha))<0$.} In this case, $h_1$ is a quadratic function with a maximum point $\Lambda_0>0$.  
	If $\alpha$ is small enough, then the discriminant of $h_1$ is positive since 
	\begin{equation*}
		\begin{split}
			&~\left(	2(\beta-\alpha) +\frac{\gamma_1\alpha\{\gamma_2\rho-\gamma_1\alpha+\gamma_2(\beta-\alpha)+2\lambda\}}{\gamma_2(\gamma_2\rho-\gamma_1\alpha+\lambda)}				\right)^2\\
			&~-4\frac{	\alpha^2\Big(	\gamma_1\alpha-\gamma_2\rho+\gamma_2(\beta-\alpha)	\Big)^2	}{\gamma_2^2(\gamma_2\rho-\gamma_1\alpha+\lambda)}\cdot\left(\frac{\gamma_1^2}{4(\gamma_2\rho-\gamma_1\alpha+\lambda)}+\frac{\lambda\gamma_1^2}{4\gamma_2\rho(\gamma_2\rho-\gamma_1\alpha+\lambda)} \right)\\
			=&~\left(	2(\beta-\alpha) +\frac{\gamma_1\alpha\{\gamma_2\rho-\gamma_1\alpha+\gamma_2(\beta-\alpha)+2\lambda\}}{\gamma_2(\gamma_2\rho-\gamma_1\alpha+\lambda)}				\right)^2\\
			&~-\frac{	\gamma_1^2\alpha^2\Big(	\gamma_1\alpha-\gamma_2\rho+\gamma_2(\beta-\alpha)	\Big)^2	}{\gamma_2^2(\gamma_2\rho-\gamma_1\alpha+\lambda)^2}-\frac{	\lambda\gamma_1^2\alpha^2\Big(	\gamma_1\alpha-\gamma_2\rho+\gamma_2(\beta-\alpha)	\Big)^2	}{\rho\gamma_2^3(\gamma_2\rho-\gamma_1\alpha+\lambda)^2}\\
			>&~4(\beta-\alpha)^2-o(\alpha)\\
			>&~0.
		\end{split}
	\end{equation*}	
	As a result, $f(\Lambda_0)-\frac{\lambda\gamma_1^2}{4\gamma_2\rho(\gamma_2\rho-\gamma_1\alpha+\lambda)}>0,$ which implies that $f(\Lambda_0)>0$.  Thus,  
	\begin{equation*}
		\begin{split}
			\widetilde {\mathcal M}_{11}=f(\Lambda_0)\alpha^2-2g(\Lambda_0)\alpha^2+\frac{\lambda\gamma_2\rho}{\gamma_2\rho-\gamma_1\alpha+\lambda}>0,\quad\widetilde {\cal M}_{22}=f(\Lambda_0)(\beta-\alpha)^2>0\\
		\end{split}
	\end{equation*}
	and
	\begin{equation*}
		\textrm{Det}[\widetilde {\cal M}]=\frac{\lambda\gamma_2\rho(\beta-\alpha)^2}{\gamma_2\rho-\gamma_1\alpha+\lambda}\left(f(\Lambda_0)-\frac{\lambda\gamma_1^2}{4\gamma_2\rho(\gamma_2\rho-\gamma_1\alpha+\lambda)}\right)+O(\alpha)>0
	\end{equation*}	
	for small $\alpha.$ Hence, in this case, too, the matrix $\widetilde{\mathcal M}$ is positive semidefinite. 
	
	\item \textbf{Case 2: $\gamma_1\alpha-\gamma_2\rho+\gamma_2(\beta-\alpha)>0.$} In this case, we put
	\begin{equation*}
		\begin{split}
			h_2(\Lambda)
			:=&~f(\Lambda)-2g(\Lambda)\\
			=&~-\frac{	\alpha^2\Big(	\gamma_1\alpha-\gamma_2\rho+\gamma_2(\beta-\alpha)	\Big)^2	}{\gamma_2^2(\gamma_2\rho-\gamma_1\alpha+\lambda)} \Lambda^2 \\
			& ~~  + \left(\frac{2\gamma_2(\gamma_2\rho-\gamma_1\alpha)(\beta-\alpha)+\gamma_1\alpha(\gamma_2\rho-\gamma_1\alpha)+\gamma_1\gamma_2\alpha(\beta-\alpha)+2\gamma_2\rho\lambda}{\gamma_2(\gamma_2\rho-\gamma_1\alpha+\lambda)} 		\right)\Lambda\\
			& ~ -\frac{\gamma_1^2}{4(\gamma_2\rho-\gamma_1\alpha+\lambda)}.\\
		\end{split}
	\end{equation*}
	The discriminant of the quadratic function $h_2$ is positive since 
	\begin{equation*}
		\begin{split}
			&~\left(\frac{2\gamma_2(\gamma_2\rho-\gamma_1\alpha)(\beta-\alpha)+\gamma_1\alpha(\gamma_2\rho-\gamma_1\alpha)+\gamma_1\gamma_2\alpha(\beta-\alpha)+2\gamma_2\rho\lambda}{\gamma_2(\gamma_2\rho-\gamma_1\alpha+\lambda)} \right)^2\\
			&~-4\frac{	\alpha^2\Big(	\gamma_1\alpha-\gamma_2\rho+\gamma_2(\beta-\alpha)	\Big)^2	}{\gamma_2^2(\gamma_2\rho-\gamma_1\alpha+\lambda)}\cdot\frac{\gamma_1^2}{4(\gamma_2\rho-\gamma_1\alpha+\lambda)}\\
			=&~\left(\frac{2\gamma_2(\gamma_2\rho-\gamma_1\alpha)(\beta-\alpha)+\gamma_1\alpha(\gamma_2\rho-\gamma_1\alpha)+\gamma_1\gamma_2\alpha(\beta-\alpha)+2\gamma_2\rho\lambda}{\gamma_2(\gamma_2\rho-\gamma_1\alpha+\lambda)} \right)^2\\
			&~-\frac{	\Big(\gamma_1\alpha	(\gamma_2\rho-\gamma_1\alpha)+\gamma_1\gamma_2\alpha(\beta-\alpha)	\Big)^2	}{\gamma_2^2(\gamma_2\rho-\gamma_1\alpha+\lambda)^2}
			> 0.
		\end{split}
	\end{equation*}	
	Let us denote the maximum point by $\Lambda_1$. Then $\Lambda_1>0$ and so $h_2(\Lambda_1)>0.$ To show that the determinant of $\widetilde M$ is positive  we first show that $$f(\Lambda_1)-\frac{\lambda\gamma_1^2}{4\gamma_2\rho(\gamma_2\rho-\gamma_1\alpha+\lambda)}>0$$ for $\alpha$ small enough. Indeed,
	\begin{equation*}\label{Lambda-1}
		\begin{split}
			&~f(\Lambda_1)-\frac{\lambda\gamma_1^2}{4\gamma_2\rho(\gamma_2\rho-\gamma_1\alpha+\lambda)}\\
			=&~\frac{	1	}{4\alpha^2\Big(	\gamma_1\alpha-\gamma_2\rho+\gamma_2(\beta-\alpha)	\Big)^2(\gamma_2\rho-\gamma_1\alpha+\lambda)}\cdot\\
			&~\Bigg(\Big\{2\gamma_2(\beta-\alpha)(\gamma_2\rho-\gamma_1\alpha)+2\lambda(\gamma_1\alpha-\gamma_2\rho+\gamma_2(\beta-\alpha))+\gamma_1\alpha(\gamma_2\rho-\gamma_1\alpha+\gamma_2(\beta-\alpha))+2\gamma_2\rho\lambda			\Big\}^2\\
			&~\quad-4\lambda^2(\gamma_1\alpha-\gamma_2\rho+\gamma_2(\beta-\alpha))^2-\alpha^2\gamma_1^2\Big(	\gamma_1\alpha-\gamma_2\rho+\gamma_2(\beta-\alpha)	\Big)^2\left(1+\frac{\lambda}{\gamma_2\rho}\right)\Bigg)\\
			>&~\frac{	1	}{4\alpha^2\Big(	\gamma_1\alpha-\gamma_2\rho+\gamma_2(\beta-\alpha)	\Big)^2(\gamma_2\rho-\gamma_1\alpha+\lambda)}\cdot\left(4\gamma_2^2\rho^2\lambda^2-o(\alpha)\right)>0\\
		\end{split}
	\end{equation*}	
	for $\alpha$ small enough. In this case, 
	\begin{equation*}
		\begin{split}
			\widetilde {\cal M}_{11}=h_2(\Lambda_1)\alpha^2+\frac{\lambda\gamma_2\rho}{\gamma_2\rho-\gamma_1\alpha+\lambda}>0,\quad\widetilde {\cal M}_{22}=f(\Lambda_1)(\beta-\alpha)^2>0\\
		\end{split}
	\end{equation*}
	and
	\begin{equation*}
		\textrm{Det}[\widetilde {\cal M}]=\frac{\lambda\gamma_2\rho(\beta-\alpha)^2}{\gamma_2\rho-\gamma_1\alpha+\lambda}\left(f(\Lambda_1)-\frac{\lambda\gamma_1^2}{4\gamma_2\rho(\gamma_2\rho-\gamma_1\alpha+\lambda)}\right)+O(\alpha)>0,
	\end{equation*}	
	and so $\widetilde{\cal M}$ is positive semidefinite. 
	\end{itemize}
	In conclusion, we can always find a constant $\Lambda>0$ such that $\widetilde {\cal M}$ is positive semidefinite. Since the terminal value $$\begin{pmatrix}	\frac{\gamma_2}{2}+\Lambda\alpha^2&\Lambda\alpha(\beta-\alpha)\\
		\Lambda\alpha(\beta-\alpha)& \Lambda(\beta-\alpha)^2		\end{pmatrix}$$ is positive definite, all the coefficients in  \eqref{eq:Riccati-tilde-app}  satisfy the requirements in \cite[Proposition 2.1, 2.2]{Kohlmann2003}. As a result, the system \eqref{eq:Riccati-tilde-app} has a unique solution in $L^\infty( [0,T];\mathbb S^2 )\cap C([0,T];\mathbb S^2)$. 
\hfill $\Box$

\begin{remark}
	In the case of risk-neutral investors, i.e.~if $\lambda=0$, the assumption that $\alpha$ is small enough can be dropped. Indeed, in the risk-neutral case, the matrix $\widetilde{\cal M}$ can be written as
	\[
	\left(\begin{matrix}
		f(\Lambda)\alpha^2 & f(\Lambda)\alpha(\beta-\alpha)\\
		f(\Lambda)\alpha(\beta-\alpha) & f(\Lambda)(\beta-\alpha)^2 
	\end{matrix}\right),
	\] 
	where
	$$f(\Lambda):=-\frac{	\alpha^2\Big(	\gamma_1\alpha-\gamma_2\rho+\gamma_2(\beta-\alpha)	\Big)^2	}{2\gamma_2^2(\gamma_2\rho-\gamma_1\alpha)} \Lambda^2  +\left(	2(\beta-\alpha)+\frac{\gamma_1\alpha}{\gamma_2} +\frac{\gamma_1\alpha(\beta-\alpha)}{(\gamma_2\rho-\gamma_1\alpha)}				\right)\Lambda -\frac{\gamma_1^2}{2(\gamma_2\rho-\gamma_1\alpha)}.$$

	Since the determinant of $\widetilde{\mathcal M}$ is zero, it is sufficient to prove that $f(\Lambda)>0$, for some $\Lambda$. This can be seen as follows. If $\alpha(\gamma_1\alpha-\gamma_2\rho+\gamma_2(\beta-\alpha))=0$, then $f(\Lambda)>0$ by choosing a $\Lambda$ large enough.  If $\alpha(\gamma_1\alpha-\gamma_2\rho+\gamma_2(\beta-\alpha))\neq 0$, then the maximum point of the quadratic function $f$ is strictly positive and the discriminant of  $f$ is positive since 
	\begin{equation*}
		\begin{split}
			&\left(	2(\beta-\alpha)+\frac{\gamma_1\alpha}{\gamma_2} +\frac{\gamma_1\alpha(\beta-\alpha)}{(\gamma_2\rho-\gamma_1\alpha)}				\right)^2-4\frac{	\alpha^2\Big(	\gamma_1\alpha-\gamma_2\rho+\gamma_2(\beta-\alpha)	\Big)^2	}{2\gamma_2^2(\gamma_2\rho-\gamma_1\alpha)}\cdot\frac{\gamma_1^2}{2(\gamma_2\rho-\gamma_1\alpha)} \\
			=&	\left(	2(\beta-\alpha)+\frac{	\gamma_1\alpha\Big(	\gamma_2\rho-\gamma_1\alpha+\gamma_2(\beta-\alpha)	\Big)	}{\gamma_2(\gamma_2\rho-\gamma_1\alpha)}	\right)^2-\frac{	\gamma_1^2\alpha^2\Big(	\gamma_1\alpha-\gamma_2\rho+\gamma_2(\beta-\alpha)	\Big)^2	}{\gamma_2^2(\gamma_2\rho-\gamma_1\alpha)^2}\\
			>&	\left(	2(\beta-\alpha)+\frac{	\gamma^2_1\alpha^2\Big(	\gamma_2\rho-\gamma_1\alpha+\gamma_2(\beta-\alpha)	\Big)^2	}{\gamma_2(\gamma_2\rho-\gamma_1\alpha)}	\right)-\frac{	\gamma_1^2\alpha^2\Big(	\gamma_2\rho-\gamma_1\alpha+\gamma_2(\beta-\alpha)	\Big)^2	}{\gamma_2^2(\gamma_2\rho-\gamma_1\alpha)^2}\\
			>&~0.
		\end{split}
	\end{equation*}	
	Hence there exists a $\Lambda>0$ such that $f(\Lambda)>0.$ 
	
\end{remark}

\section{The verification theorem}\label{sec-ver}

To verify that the strategy given by equations \eqref{opt1} and \eqref{opt2}  is indeed optimal we first prove that the cost functional can be written as a sum of two complete square terms and a correction term that we will identify as the value function.  

\begin{proposition}\label{lemma:verification}
	Let the \textbf{standing assumption} hold.	Then the cost functional can be rewritten as
	\begin{equation*}
		\begin{split}
			J(t, Z)
			=&~\mathbb E\Big[\int_{t}^{T}\frac{1}{\tilde a}\left(I^A_s(\mathcal X_s-\bar{\mu}_{s})\right)^2\,ds+\int_{t}^{T}\frac{1}{a}\left(I^B_s\bar{\mu}_{s}+I^D_s\right)^2\,ds\Big]\\
			&~+\textnormal{Var}({\mu}_{t-})(A_{t})+\bar{\mu}_{t-}^\top B_t\bar{\mu}_{t-}+D^\top_t\bar{\mu}_{t-}+\mathbb E[F_t],
		\end{split}	
	\end{equation*}
where $\mu_\cdot$ is the law of $\mathcal X_\cdot$.
	In particular, the cost functional reaches its global minimum if  
	\[
	\int_{t}^{T}\left(I^A_s(\mathcal X_s-\bar{\mu}_{s})\right)^2\,ds=\int_{t}^{T}\left(I^B_s\bar{\mu}_{s}+I^D_s\right)^2\,ds=0,\quad a.s.
	\]
	In this case, the value function is indeed given by \eqref{value2}.
\end{proposition}
\begin{proof}
For any strategy $Z\in\mathscr A_t$,  we first separate the cost of the jump at the terminal time from the cost functional. To this end, we write the cost functional as 
\begin{equation}\label{eq:cost-rewritten-1}
	\begin{split}
		J(t,Z)=&~	\mathbb E\left[ \int_{[t,T)} \left( Y_{s-}\,dZ_s+\frac{\gamma_{2}}{2}\,d[Z]_s +\sigma_s\,d[Z,W]_s \right)	+\int_t^T\lambda X^2_s\,ds	 \right]\\
		&~+ \mathbb E\left[-\left(	Y_{T-}-\frac{\gamma_2}{2}\Delta X_T	\right)\Delta X_T \right]\\
		=&~\mathbb E\left[	\int_{[t,T)} \left( Y_{s-}\,dZ_s+\frac{\gamma_{2}}{2}\,d[Z]_s+\sigma_s\,d[Z,W]_s \right)	+\int_t^T\lambda  X^2_s\,ds	 \right]\\
		&~+\mathbb E\left[ \frac{\gamma_2}{2} X^2_{T-} +X_{T-}Y_{T-}	\right] \quad (\textrm{ since }X_T=0)\\
		=&~\mathbb E\left[	\int_{[t,T)} \left( Y_{s-}\,dZ_s+\frac{\gamma_{2 }}{2}\,d[Z]_s+\sigma_s\,d[Z,W]_s \right)	+\int_t^T\lambda X^2_s\,ds	 \right]\\
		&~+\mathbb E\left[(\mathcal X_{T-}-\bar{\mu}_{T-})^\top A_{T}(\mathcal X_{T-}-\bar{\mu}_{T-})+\bar{\mu}_{T-}^\top B_T\bar{\mu}_{T-}+D^\top_T\bar{\mu}_{T-}+F_T \right].
	\end{split}
\end{equation}
Next, we are going to analyze the expected jump cost term-by-term. From this we will see that many terms cancel and then arrive at the desired representation of the cost functional.
\begin{itemize}
	\item We start with the term $\mathbb E\left[ \mathcal X_{T-} A_{T} \mathcal X_{T-} \right]$. Using It\^o's formula in \cite[Theorem 36]{Protter-2005}, 
\begin{equation*}
	\begin{split}
\int_{\mathbb R^3}x^\top A_{T}x\mu_{T-}(dx) =&~\mathbb E\Big[\int_{t}^{T-}2(A_s\mathcal X_{s-})^\top d\mathcal X_s+Tr(A_sd[\mathcal X,\mathcal X]^c_s)\\
&~+\sum_{t\leq s< T}\left(\mathcal X_s^\top A_s\mathcal X_s-\mathcal X_{s-}^\top A_s\mathcal X_{s-}-2(A_s\mathcal X_{s-})^\top\Delta\mathcal X_s\right)\Big]\\
&~+\int_{\mathbb R^3}x^\top A_{t}x\mu_{t-}(dx)+\int_t^T\int_{\mathbb R^3}x^\top
\dot{A}_{s}x\mu_{s}(dx)\,ds.
\end{split}
\end{equation*}
Note that
\begin{equation*}
	\begin{split}
		d[\mathcal X,\mathcal X]^c_s=\begin{pmatrix}
			d[Z^c,Z^c]_s &  -\sigma_s\,d[Z^c,W]_s-\gamma_2d[Z^c,Z^c]_s & 0 \\
			-\sigma_sd[Z^c,W]_s-\gamma_2d[Z^c,Z^c]_s & \sigma^2_s\,ds+\gamma_2^2\,d[Z^c,Z^c]_s+2\gamma_2\sigma_s\,d[Z^c,W]_s & 0\\
			0&0&0
			\end{pmatrix},
	\end{split}
\end{equation*}
which implies by the relationship between the entries of the matrix $A$ (cf.~the statement above \eqref{eq:tilde-A}) that
\[
	\int_{t}^{T-}\textnormal{Tr}( A_s\,d[\mathcal X^c,\mathcal X^c]_s )\,ds=\int_t^T\left( -\frac{\gamma_2}{2}\,d[Z^c,Z^c]_s-\sigma_s\,d[Z^c,W]_s+\sigma^2_s A_{22,s}\,ds\right).
\]

Taking this back into the above equation shows that
\begin{equation}\label{eq:E[XAX]}
	\begin{split}
& \int_{\mathbb R^3}x^\top A_{T}x\mu_{T-}(dx) \\ =&~\mathbb E\left[\int_{t}^{T}2(A_s\mathcal X_s)^\top\left(\mathcal H\mathcal X_s+\overline{\mathcal H}\mathbb E[\mathcal X_s]+\mathcal G\right)\,ds+\int_{t}^{T-}2(A_s\mathcal X_{s-})^\top\mathcal K\,dZ_s\right]\\
&~+\mathbb E\left[\int_t^T\left( -\frac{\gamma_2}{2}\,d[Z^c,Z^c]_s-\sigma_s\,d[Z^c,W]_s+\sigma^2_s A_{22,s}\,ds\right)\right]\\
&~+\mathbb E\left[\sum_{t\leq s< T}\left(\mathcal X_s^\top A_s\mathcal X_s-\mathcal X_{s-}^\top A_s\mathcal X_{s-}-2(A_s\mathcal X_{s-})^\top\mathcal K\Delta Z_s\right)\right]\\
&~+\int_{\mathbb R^3}x^\top A_{t}x\mu_{t-}(dx)+\mathbb E\left[\int_t^T\mathcal X_s^\top
\dot{A}_{s}\mathcal X_s\,ds\right].
	\end{split}
\end{equation}

\item Next, we consider the term $\bar{\mu}_{T-}^\top A_{T}\bar{\mu}_{T-} $. In view of \eqref{state_matrix} the expected value $\overline\mu$ follows the dynamics
\begin{equation}\label{ODE:mu-bar}
		d\overline\mu_s=\Big(	(\mathcal H+\overline{\mathcal H})\overline\mu_s	+\mathcal G	\Big)\,ds+\mathcal K\,d\mathbb E[Z_s].
	\end{equation}
Applying the chain rule to $\overline\mu^\top A\overline \mu$ from $t-$ to $T-$, it follows that
\begin{equation*}
	\begin{split}
	\bar{\mu}_{T-}^\top A_{T}\bar{\mu}_{T-} =&~\int_{t}^{T-}2(A_s\bar{\mu}_{s-})^\top d\overline\mu_s+ 
	\sum_{t\leq s< T}\left(\bar{\mu}_s^\top A_s\bar{\mu}_s-\bar{\mu}_{s-}^\top A_s\bar{\mu}_{s-}-2(A_s\bar{\mu}_{s-})^\top\Delta\bar{\mu}_s\right)\\
		&+\bar{\mu}_{t-}^\mathrm TA_{t}\bar{\mu}_{t-} +\int_t^{T}\bar{\mu}_{s}^\mathrm T\dot{A}_{s}\bar{\mu}_{s}\,ds.
\end{split}
\end{equation*}

Taking \eqref{ODE:mu-bar} into the expression of $\overline\mu^\top_{T-}A_T\overline\mu_{T-}$ we arrive at
\begin{equation}\label{eq:mu-A-mu}
	\begin{split}		\overline\mu^\top_{T-}A_T\overline\mu_{T-}
		=&~ \bar{\mu}_{t-}^\top A_{t}\bar{\mu}_{t-} + \int_{t}^{T}2(A_s\bar{\mu}_s)^\top\left(\mathcal H\overline\mu_s+\overline{\mathcal H}\overline\mu_s+\mathcal G\right)\,ds+\int_{t}^{T-}2(A_s\bar{\mu}_{s-})^\top\mathcal K\,d\mathbb E[Z_s] \\
		&~+\sum_{t\leq s< T}\Big(\bar{\mu}_s^\top A_s\bar{\mu}_s-\bar{\mu}_{s-}^\top A_s\bar{\mu}_{s-}-2(A_s\bar{\mu}_{s-})^\top\Delta\overline\mu_s\Big) +\int_t^{T}\bar{\mu}_{s}^\top\dot{A}_{s}\bar{\mu}_{s}\,ds.
	\end{split}
\end{equation}

\item Similarly to the last step it holds that
\begin{equation}\label{eq:mu-B-mu}
	\begin{split}		\overline\mu^\top_{T-}B_T\overline\mu_{T-}
		=& ~ \bar{\mu}_{t-}^\top B_{t}\bar{\mu}_{t-} + \int_{t}^{T}2(B_s\bar{\mu}_s)^\top\left(\mathcal H\overline\mu_s+\overline{\mathcal H}\overline\mu_s+\mathcal G\right)\,ds+\int_{t}^{T-}2(B_s\bar{\mu}_{s-})^\top\mathcal K\,d\mathbb E[Z_s] \\
		&~+\sum_{t\leq s< T}\Big(\bar{\mu}_s^\top B_s\bar{\mu}_s-\bar{\mu}_{s-}^\top B_s\bar{\mu}_{s-}-2(B_s\bar{\mu}_{s-})^\top\Delta\overline\mu_s\Big) +\int_t^{T}\bar{\mu}_{s}^\top\dot{B}_{s}\bar{\mu}_{s}\,ds.
	\end{split}
\end{equation}

\item Applying the chain rule to $D^\top\overline\mu$, we see that
\begin{equation}\label{eq:Dmu-F}
	\begin{split}
		D_{T}^\top \bar{\mu}_{T-} 
		=&~\mathbb E\Big[\int_{t}^{T}D_s^\top\left(\mathcal H\mathcal X_s+\overline{\mathcal H}\mathbb E[\mathcal X_s]+\mathcal G\right)\,ds+\int_{t}^{T-}D_s^\top\mathcal K\,dZ_s\\
		&~+\sum_{t\leq s< T}\left((D_s^\top\mathcal X_s-D_s^\top\mathcal X_{s-})-D_s^\top\Delta \mathcal X_s\right)\Big]+D_{t}^\top\bar{\mu}_{t-} +\int_t^{T}\dot{D}^\top_{s}\bar{\mu}_{s}\,ds \\	
	     =&~\mathbb E\left[\int_{t}^{T}D_s^\top\left(\mathcal H\mathcal X_s+\overline{\mathcal H}\mathbb E[\mathcal X_s]+\mathcal G\right)\,ds+\int_{t}^{T-}D_s^\top\mathcal K\,dZ_s +D_{t}^\top\bar{\mu}_{t-} +\int_t^{T}\dot{D}^\top_{s}\bar{\mu}_{s}\,ds \right].
	\end{split}
\end{equation}
\end{itemize}

Next, we collect all the terms in \eqref{eq:E[XAX]}, \eqref{eq:mu-A-mu}, \eqref{eq:mu-B-mu} and \eqref{eq:Dmu-F} involving jumps. Their sum equals 
	\begin{equation}\label{jumpterms}
	\begin{split}
	&~\mathbb E\Big[\int_{t}^{T-}2(A_s\mathcal X_{s-})^\top\mathcal K\,dZ_s+\sum_{t\leq s< T}\left(\mathcal X_s^\top A_s\mathcal X_s-\mathcal X_{s-}^\top A_s\mathcal X_{s-}-2(A_s\mathcal X_{s-})^\top\mathcal K\Delta Z_s\right)\Big]\\
	&~-\mathbb E\Big[\int_{t}^{T-}2(A_s\bar{\mu}_{s-})^\top\mathcal K\,dZ_s+\sum_{t\leq s< T}\left(\bar{\mu}_s^\top A_s\bar{\mu}_s-\bar{\mu}_{s-}^\top A_s\bar{\mu}_{s-}-2(A_s\bar{\mu}_{s-})^\top\Delta\mathcal X_s\right)\Big]\\
	&~+\mathbb E\Big[\int_{t}^{T-}2(B_s\bar{\mu}_{s-})^\top \mathcal K\,dZ_s+\sum_{t\leq s< T}\left(\bar{\mu}_s^\top B_s\bar{\mu}_s-\bar{\mu}_{s-}^\top B_s\bar{\mu}_{s-}-2(B_s\bar{\mu}_{s-})^\top\Delta\mathcal X_s\right)\Big]\\
	&~+\mathbb E\left[\int_{t}^{T-}D_s^\top\mathcal K\,dZ_s\right].
	\end{split}	
\end{equation}

Since 
	\begin{equation}\label{equality:KA}
	\begin{split}
\mathcal K^\top A=\begin{pmatrix}
	-1&\gamma_2&0
\end{pmatrix}A=\begin{pmatrix}
0&-\frac{1}{2}&0
\end{pmatrix}=\mathcal K^\top B,
	\end{split}	
\end{equation}
we obtain that
\[
-\mathbb E\Big[\int_{t}^{T-}2(A_s\bar{\mu}_{s-})^\top\mathcal K\,dZ_s\Big]+\mathbb E\Big[\int_{t}^{T-}2(B_s\bar{\mu}_{s-})^\top\mathcal K\,dZ_s\Big]=0.
\]

Since $\Delta\bar{\mu}=\mathbb E[\mathcal K\Delta Z]$ we also obtain that
	\begin{equation*}
	\begin{split}
\mathbb E\left[(\bar{\mu}_s^\top B_s\bar{\mu}_s-\bar{\mu}_{s-}^\top B_s\bar{\mu}_{s-})-2(B_s\bar{\mu}_{s-})^\top \Delta\mathcal X_s\right]=&~\Delta\bar{\mu}_s^\top B_s\Delta\bar{\mu}_s\\
=&~\mathbb E[\Delta Z_s\mathcal K^\top B_s\Delta\bar{\mu}_s]\\
=&~\mathbb E[\Delta Z_s\mathcal K^\top A_s\Delta\bar{\mu}_s]\\
=&~\Delta\bar{\mu}_s^\top A_s\Delta\bar{\mu}_s\\
=&~\mathbb E\left[(\bar{\mu}_s^\top A_s\bar{\mu}_s-\bar{\mu}_{s-}^\top A_s\bar{\mu}_{s-})-2(A_s\bar{\mu}_{s-})^\top\Delta\mathcal X_s\right].
	\end{split}	
\end{equation*}
Using \eqref{equality:KA} again, we have that
	\begin{equation*}
	\begin{split}
2(A_s\mathcal X_{s-})^\top\mathcal K\,dZ_s & =2\mathcal K^\top A_s\mathcal X_{s-}\,dZ_s \\ & = \begin{pmatrix}
	0&-1&0 \end{pmatrix}\mathcal X_{s-}\,dZ_s \\ & =-Y_{s-}\,dZ_s.
	\end{split}
	\end{equation*}	
Moreover, the definition of $\mathcal K$ implies that
	\begin{equation*}
	\begin{split}
\mathcal X_s^\top A_s\mathcal X_s-\mathcal X_{s-}^\top A_s\mathcal X_{s-}-2(A_s\mathcal X_{s-})^\top\mathcal K\Delta Z_s & =\Delta\mathcal X_s^\top A_s\Delta\mathcal X_s \\ & =\Delta Z_s\mathcal K^\top A_s\Delta\mathcal X_s  \\ & =-\frac{1}{2}\Delta Z_s\Delta Y_s  \\ & =-\frac{\gamma_2}{2}(\Delta Z_s)^2,
	\end{split}	
\end{equation*}
and that
	\begin{equation*}
	\begin{split}
D^\top\mathcal K & =\begin{pmatrix} 
	D_1&D_2&D_3
\end{pmatrix}\begin{pmatrix}
	-1\\\gamma_2\\0
\end{pmatrix}=-D_1+\gamma_2D_2=0.\\
	\end{split}	
\end{equation*}

As a result, the jump terms \eqref{jumpterms} together with the term $\mathbb E\left[\int_t^T\left( -\frac{\gamma_2}{2}\,d[Z^c,Z^c]_s-\sigma_s\,d[Z^c,W]_s\right)\right]$ in \eqref{eq:E[XAX]} cancel with the first three terms in  \eqref{eq:cost-rewritten-1}.
Hence, taking \eqref{eq:E[XAX]}, \eqref{eq:mu-A-mu}, \eqref{eq:mu-B-mu} and \eqref{eq:Dmu-F} into \eqref{eq:cost-rewritten-1} yields that
\begin{equation*}
	\begin{split}
	& ~ J(t,Z) \\ =&~\mathbb E\left[\int_{t}^{T}2(A_s\mathcal X_s)^\top\left(\mathcal H\mathcal X_s+\overline{\mathcal H}\mathbb E[\mathcal X_s]+\mathcal G\right)\,ds+\int_{t}^{T}\mathcal D^\top_s A_s\mathcal D_s\,ds+\int_t^T\mathcal X_s^\top
		\dot{A}_{s}\mathcal X_s\,ds+\int_t^T\mathcal X^\top_s\mathcal Q \mathcal X_s\,ds\right]\\
		&~-\mathbb E\left[\int_{t}^{T}2(A_s\bar{\mu}_s)^\top\left(\mathcal H\mathcal X_s+\overline{\mathcal H}\mathbb E[\mathcal X_s]+\mathcal G\right)\,ds\right]-\int_t^{T}\bar{\mu}_{s}^\top\dot{A}_{s}\bar{\mu}_{s}\,ds \\
		&~+\mathbb E\left[\int_{t}^{T}2(B_s\bar{\mu}_s)^\top \left(\mathcal H\mathcal X_s+\overline{\mathcal H}\mathbb E[\mathcal X_s]+\mathcal G\right)\,ds\right]+\int_t^{T}\bar{\mu}_{s}^\top\dot{B}_{s}\bar{\mu}_{s}\,ds \\
		&~+\mathbb E\left[\int_{t}^{T}D_s^\top\left(\mathcal H\mathcal X_s+\overline{\mathcal H}\mathbb E[\mathcal X_s]+\mathcal G\right)\,ds\right] +\int_t^{T}\dot{D}^\top_{s}\bar{\mu}_{s}\,ds +\mathbb E\left[\int^T_t\dot{F}_s\,ds\right]\\
		&~+\textnormal{Var}(\mu_{t-})(A_t)+\overline\mu^\top_{t-}B_t\overline\mu_{t-}+D^\top_t\overline\mu_{t-}+\mathbb E[F_t],
		\end{split}
	\end{equation*}
where
\[
	\mathcal Q:=\begin{pmatrix}
		\lambda & 0 & 0\\
		0 & 0 & 0\\
		0& 0 & 0
	\end{pmatrix}.
\]

Recalling that $\overline\mu=\mathbb E[\mathcal X]$ and collecting the terms $\mathcal X^\top(\cdots)\mathcal X$, $\overline\mu^\top(\cdots)\overline\mu$, $(\cdots)\overline\mu$ and other terms, we have that
\begin{equation*}
	\begin{split}
		& ~ J(t,Z) \\
		=&~\mathbb E\left[\int_{t}^{T}\mathcal X_s^\top\left(\mathcal Q+2\mathcal H^\top A_s+\dot{A}_{s}\right)\mathcal X_s\,ds+\int_{t}^{T}\bar{\mu}_{s}^\top\left(-\dot{A}_{s}-2\mathcal H^\top A_s+\dot{B}_{s}+2\mathcal H^\top B_s+2\overline{\mathcal H}^\top B_s\right)\bar{\mu}_{s}\,ds\right.\\
		&~\left.+\int_{t}^{T}\left(2\mathcal G^\top B_s+\dot{D}_{s}^\top+D^\top_s\mathcal H +D^\top_s\overline{\mathcal H}\right)\bar{\mu}_{s}\,ds+\int_{t}^{T}\left(\mathcal D^\top_s A_s\mathcal D_s+D_s^\top\mathcal G+\dot{F}_s\right)\,ds\right]\\
		&~+\textnormal{Var}(\mu_{t-})(A_t)+\overline\mu^\top_{t-}B_t\overline\mu_{t-}+D^\top_t\overline\mu_{t-}+\mathbb E[F_t]\\
		=&~\mathbb E\left[\int_{t}^{T}(\mathcal X_s-\bar{\mu}_{s})^\top\left(\mathcal Q+\mathcal H^\top A_s+A_s\mathcal H+\dot{A}_{s}\right)(\mathcal X_s-\bar{\mu}_{s})\,ds\right.\\
		&~\left.+\int_{t}^{T}\bar{\mu}_{s}^\top\left(\mathcal Q+\dot{B}_{s}+\mathcal H^\top B_s+B_s\mathcal H+\overline{\mathcal H}^\top B_s+B_s\overline{\mathcal H}\right)\bar{\mu}_{s}\,ds\right.\\
		&~\left.+\int_{t}^{T}\left(2\mathcal G^\top B_s+\dot{D}_{s}^\top+D^\top_s\mathcal H +D^\top_s\overline{\mathcal H}\right)\bar{\mu}_{s}\,ds+\int_{t}^{T}\left(\mathcal D^\top_s A_s\mathcal D_s+D_s^\top\mathcal G+\dot{F}_s\right)\,ds\right]\\
		&~+\textnormal{Var}(\mu_{t-})(A_t)+\overline\mu^\top_{t-}B_t\overline\mu_{t-}+D^\top_t\overline\mu_{t-}+\mathbb E[F_t].
				\end{split}
	\end{equation*}

By \eqref{eq:A}, \eqref{eq:B}, \eqref{eq:D} and \eqref{eq:F} in the appendix it holds that 
\begin{equation*}
	\begin{split}
		\frac{ (I^A)^\top I^A }{\tilde a} & = ~\mathcal Q+\mathcal H^\top A+A\mathcal H+\dot{A}\\[2mm]
		\frac{(I^B)^\top I^B}{a}& = ~\mathcal Q+\dot{B}+\mathcal H^\top B+B\mathcal H+\overline{\mathcal H}^\top B+B\overline{\mathcal H}\\[2mm]
		\frac{2I^DI^B}{a} &= ~2\mathcal G^\top B_s+\dot{D}_{s}^\top+D^\top_s\mathcal H +D^\top_s\overline{\mathcal H}\\[2mm]
		\frac{(I^D)^2}{a} & = ~\mathcal D^\top A_s\mathcal D+D^\top\mathcal G+\dot{F}
	\end{split}
\end{equation*}
which gives us the desired result. 
\end{proof}
Let $\widetilde{\mathcal X}$ be the state process driven by the strategy $\widetilde Z$ given by \eqref{opt1} and \eqref{opt2}. The next theorem shows that $\widetilde{\mathcal X}$ satisfies the equality in Proposition \ref{lemma:verification} thereby concluding the verification argument. 

\begin{theorem}
	The state process $\widetilde{\mathcal X}$ driven by the strategy \eqref{opt1} and \eqref{opt2} satisfies
	\begin{equation}\label{equality-optimal-state}
			I^A(\widetilde{\mathcal X}_s-\mathbb E[\widetilde{\mathcal X}_{s}])=0,\quad I^B\mathbb E[\widetilde{\mathcal X}_{s}]+I^D=0,\qquad s\in[t,T). 
	\end{equation}
\end{theorem}
\begin{proof}
{It is straightforward to verify that $\widetilde Z$ determined by \eqref{opt1} and \eqref{opt2} is admissible.} We first prove that \eqref{equality-optimal-state} holds at $s=t$. From the definition of $\Delta \widetilde Z_t$ in \eqref{opt1}, we know that 
	\[
		\mathbb E[\Delta \widetilde Z_t]=-\frac{I^B}{a}\mathbb E[\mathcal X]-\frac{I^D}{a}, \quad
		\Delta \widetilde Z_t=-\frac{I^A}{\tilde a}(\mathcal X-\mathbb E[\mathcal X])+\mathbb E[\Delta  \widetilde Z_t],
	\]	
	where we recall that $\mathcal X:=\mathcal X_{t-}$.
	Since $I^A\mathcal K=\tilde a$ and $I^B\mathcal K=a,$ we have that
\begin{equation*}
	I^B_t\mathbb E[\widetilde{\mathcal X}_{t}]+I^D_t=I^B_t\mathbb E[\mathcal X+\mathcal K\Delta \widetilde Z_t]+I^D_t = I^B_t\mathbb E[\mathcal X]+a\mathbb E[\Delta \widetilde Z_t]+I^D_t=0
\end{equation*}
and
\begin{equation*}		
	I^A_t(\widetilde{\mathcal X}_t-\mathbb E[\widetilde{\mathcal X}_{t}])=I^A_t(\mathcal X+\mathcal K\Delta \widetilde Z_t-\mathbb E[\mathcal X+\mathcal K\Delta \widetilde Z_t])=I^A_t(\mathcal X-\mathbb E[\mathcal X])+\tilde a(\Delta \widetilde Z_t-\mathbb E[\Delta \widetilde Z_t])=0.
\end{equation*}

Next, we prove that $d\Big\{I^A_s(\widetilde{\mathcal X}_s-\mathbb E[\widetilde{\mathcal X}_{s}])\Big\}=d\Big\{I^B_s\mathbb E[\widetilde{\mathcal X}_{s}]+I^D_s\Big\}=0$ for all $s\in[t,T)$ from which \eqref{equality-optimal-state} follows. In fact, the state dynamics gives us that
\begin{equation*}\label{expectation-X}
	\begin{split}
		d\mathbb E[\widetilde{\mathcal X}_s]
		=&~\left((\mathcal H+\overline{\mathcal H})\mathbb E[\widetilde{\mathcal X}_s]+\mathcal G-\mathcal K\frac{\dot I^B_s}{a}\mathbb E[\widetilde{\mathcal X}_s]-\mathcal K\frac{\dot I^D_s}{a}-\mathcal K\frac{I^B_s}{a}(\mathcal H+\overline{\mathcal H})\mathbb E[\widetilde{\mathcal X}_s])-\mathcal K\frac{I^B_s}{a}\mathcal G\right)\,ds,
	\end{split}
\end{equation*}
and
\begin{equation*}
	\begin{split}
		d(\widetilde{\mathcal X}_s-\mathbb E[\widetilde{\mathcal X}_s])
		=&~\left(\mathcal H(\widetilde{\mathcal X}_s-\mathbb E[\widetilde{\mathcal X}_s])-\mathcal K\frac{\dot I^A_s}{\tilde a}(\widetilde{\mathcal X}_s-\mathbb E[\widetilde{\mathcal X}_s])-\mathcal K\frac{I^A_s}{\tilde a}\mathcal H(\widetilde{\mathcal X}_s-\mathbb E[\widetilde{\mathcal X}_s])\right)\,ds-\mathcal K\frac{I^A_s}{\tilde a}\mathcal D_s\,dW_s+\mathcal D_s\,dW_s.
	\end{split}
\end{equation*}
Hence, the desired result follows from the following equalities:
\begin{equation*}
	\begin{split}
	d	\Big\{I^B_s\mathbb E[\widetilde{\mathcal X}_{s}]+I^D_s\Big\}=&~\dot I^B_s\mathbb E[\widetilde{\mathcal X}_{s}]\,ds+\dot I^D_s\,ds+I^B_sd\mathbb E[\widetilde{\mathcal X}_{s}]\\
		=&~\left(\dot I^B_s+I^B_s\left((\mathcal H+\overline{\mathcal H})-\mathcal K\frac{\dot I^B_s}{a}-\mathcal K\frac{I^B_s}{a}(\mathcal H+\overline{\mathcal H})\right)\right)\mathbb E[\widetilde{\mathcal X}_{s}]\,ds\\
		&~+\dot I^D_s\,ds+I^B_s\left(\mathcal G-\mathcal K\frac{I^B_s}{a}\mathcal G-\mathcal K\frac{\dot I^D_s}{a}\right)\,ds\\
		=&~\left(\dot I^B_s+I^B_s(\mathcal H+\overline{\mathcal H})-\dot I^B_s-I^B_s(\mathcal H+\overline{\mathcal H})\right)\mathbb E[\widetilde{\mathcal X}_{s}]\,ds\\
		&~+(\dot I^D_s+I^B_s\mathcal G-I^B_s\mathcal G-\dot I^D_s)\,ds\\
		=&~0,
	\end{split}
\end{equation*}
and
\begin{equation*}
	\begin{split}
		d\Big\{I^A_s(\widetilde{\mathcal X}_s-\mathbb E[\widetilde{\mathcal X}_{s}])\Big\}=&~dI^A_s(\widetilde{\mathcal X}_s-\mathbb E[\widetilde{\mathcal X}_{s}])+I^A_sd(\widetilde{\mathcal X}_s-\mathbb E[\widetilde{\mathcal X}_{s}])\\
		=&~\left(\dot I^A_s+I^A_s\left(\mathcal H-\mathcal K\frac{\dot I^A_s}{\tilde a}-\mathcal K\frac{I^A_s}{\tilde a}\mathcal H\right)\right)(\widetilde{\mathcal X}_s-\mathbb E[\widetilde{\mathcal X}_{s}])\,ds+I^A_s\left(-\mathcal K\frac{I^A_s}{\tilde a}\mathcal D_s+\mathcal D_s\right)\,dW_s\\
		=&~\left(\dot I^A_s+I^A_s\mathcal H-\dot I^A_s-I^A_s\mathcal H\right)(\widetilde{\mathcal X}_s-\mathbb E[\widetilde{\mathcal X}_{s}])\,ds+\left(-I^A_s\mathcal D_s+I^A_s\mathcal D_s\right)\,dW_s\\
		=&~0.
	\end{split}
\end{equation*}
\end{proof}


\section{Numerical examples}\label{sec-num}

This section provides numerical simulations that illustrate the dependence of the optimal inventory process on various model parameters. In all cases $x=1$, $c=y=0$, $\sigma=0.8$ and $T=1$. 
All trajectories were generated from the same Brownian path to guarantee that the trajectories are comparable. In addition to our optimal solution we display the optimal solution in the Obizhaeva-Wang model \cite{OW-2013}, which is the canonical reference point for our model. Setting $\alpha=\beta=\sigma=0$ our model reduces to the Obizhaeva-Wang model.  

Figure 1 displays the optimal position for two extreme choice of the market risk parameter. When the investor is highly risk averse the optimal holding in our models is relatively close to the one in the Obizhaeva-Wang model with added twist that in our model the investor may take short positions generating additional sell child order flow and buy the stock back while benefitting from the additional sell order flow. Overselling own positions should not be viewed as a fraudulent attempt to manipulate prices. Instead, the large investor rationally anticipates his/her impact on future order flow when making trading decisions and uses it to his/her advantage. Similar effects have previously been observed in the literature; see \cite{FHX1} and references therein for a more detailed discussion of different manipulation strategies in portfolio liquidation models. 

\begin{figure}[h]\label{fig1}
	\begin{minipage}[c]{0.5\textwidth}
		\centering
		\includegraphics[height=5cm,width=7.5cm]{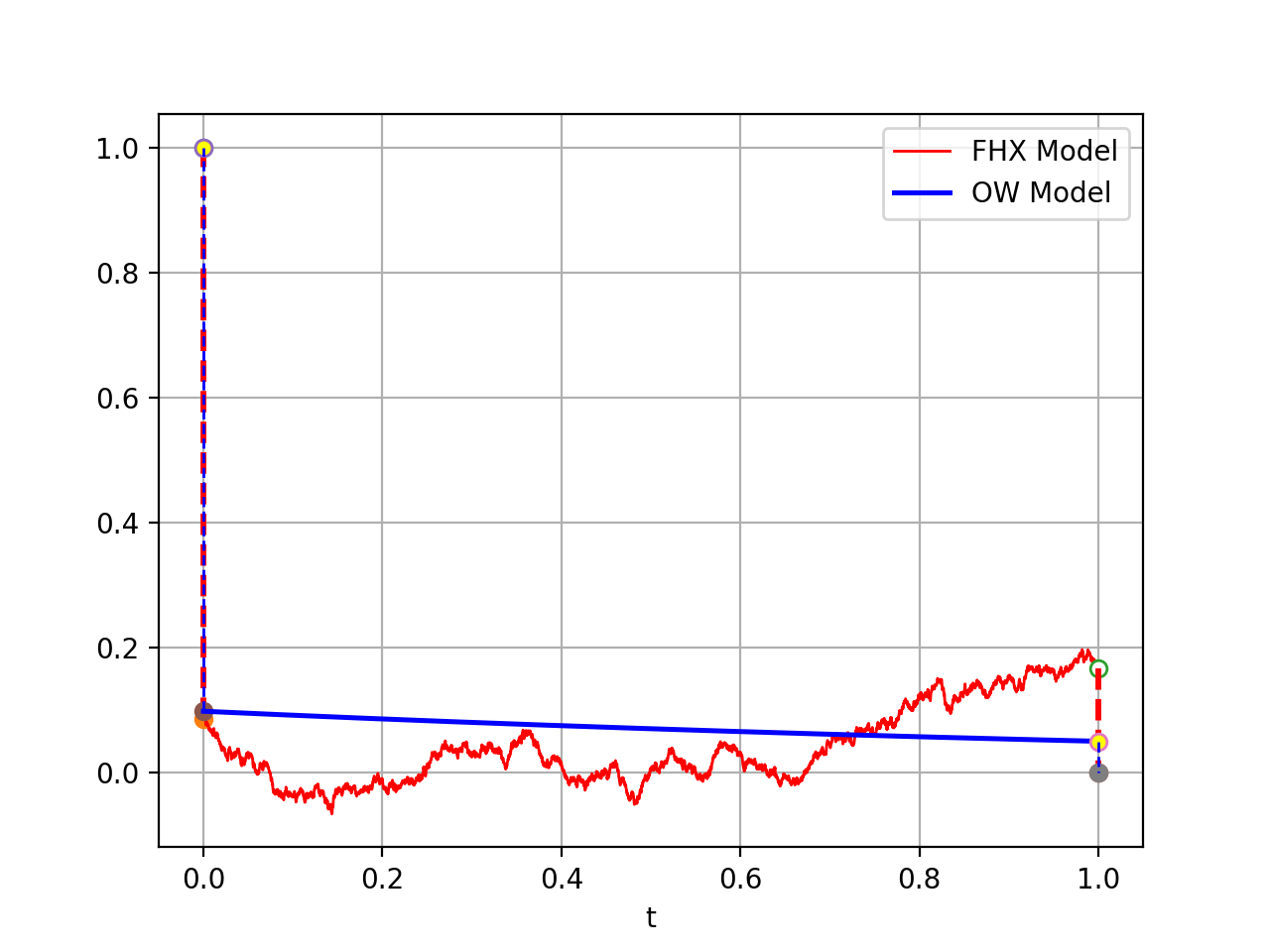}
	\end{minipage}
	\begin{minipage}[c]{0.5\textwidth}
		\centering
		\includegraphics[height=5cm,width=7.5cm]{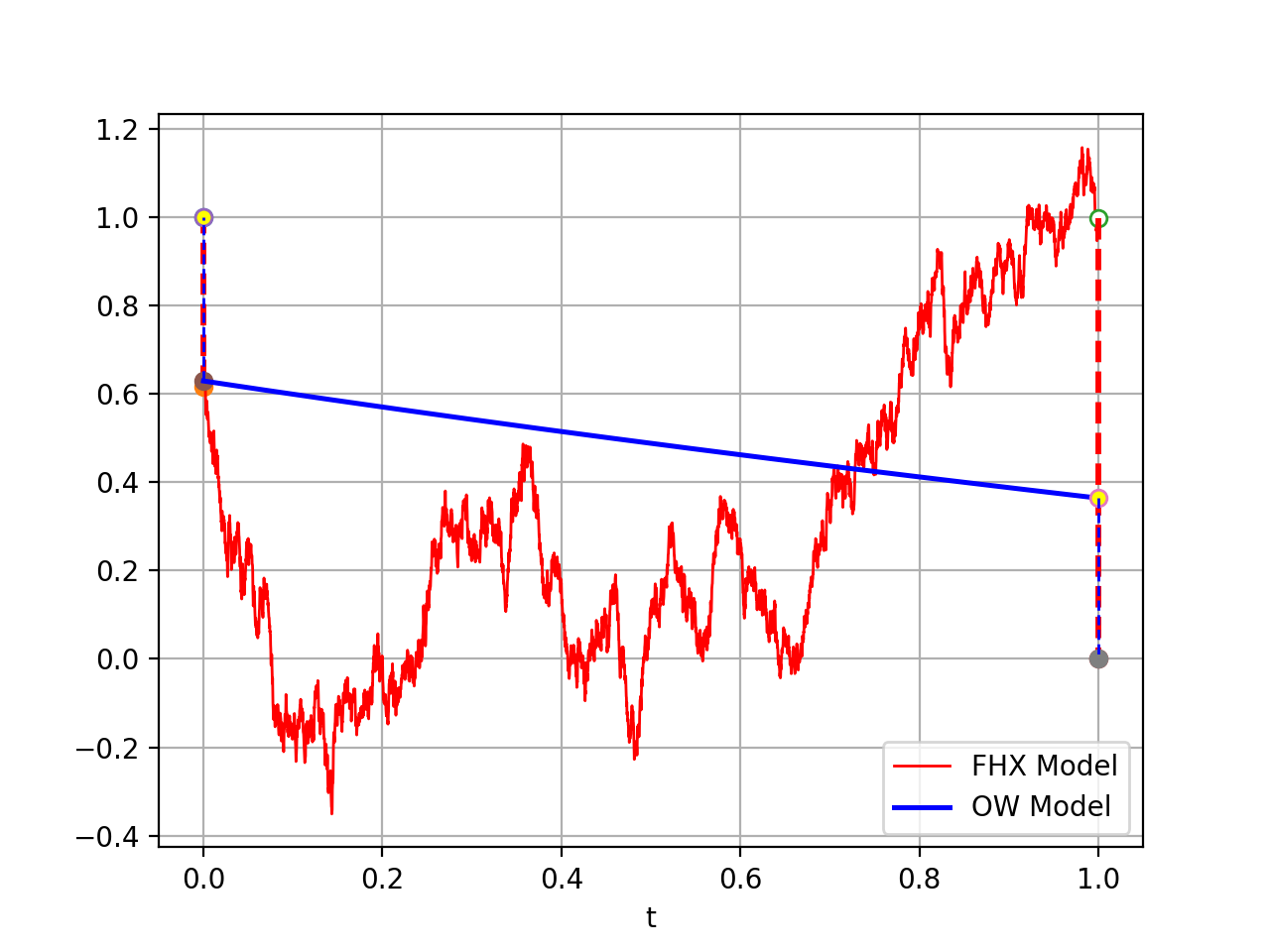}
	\end{minipage}
	\caption{Dependence of the optimal position on the risk parameter $\lambda$ for $\lambda = 1.5$ (left) and $\lambda = 0$(right). Other parameters are chosen as $\rho=0.7$, $\gamma_1 = 0.1$, $\gamma_2 = 0.5$, $\alpha = 0.5$, $\beta=1.1$. }
\end{figure}

When the investor is risk-neutral, then the variations in the optimal inventory process are much larger; portfolio holdings range from about $-0.4$ to $1.2$. This, too is very intuitive. Large portfolio holdings are  much ``cheaper'' for a risk-neutral than a risk-averse trader.

\begin{figure}[h]\label{fig2}
	\begin{minipage}[c]{0.5\textwidth}
		\centering
		\includegraphics[height=4.5cm,width=7.5cm]{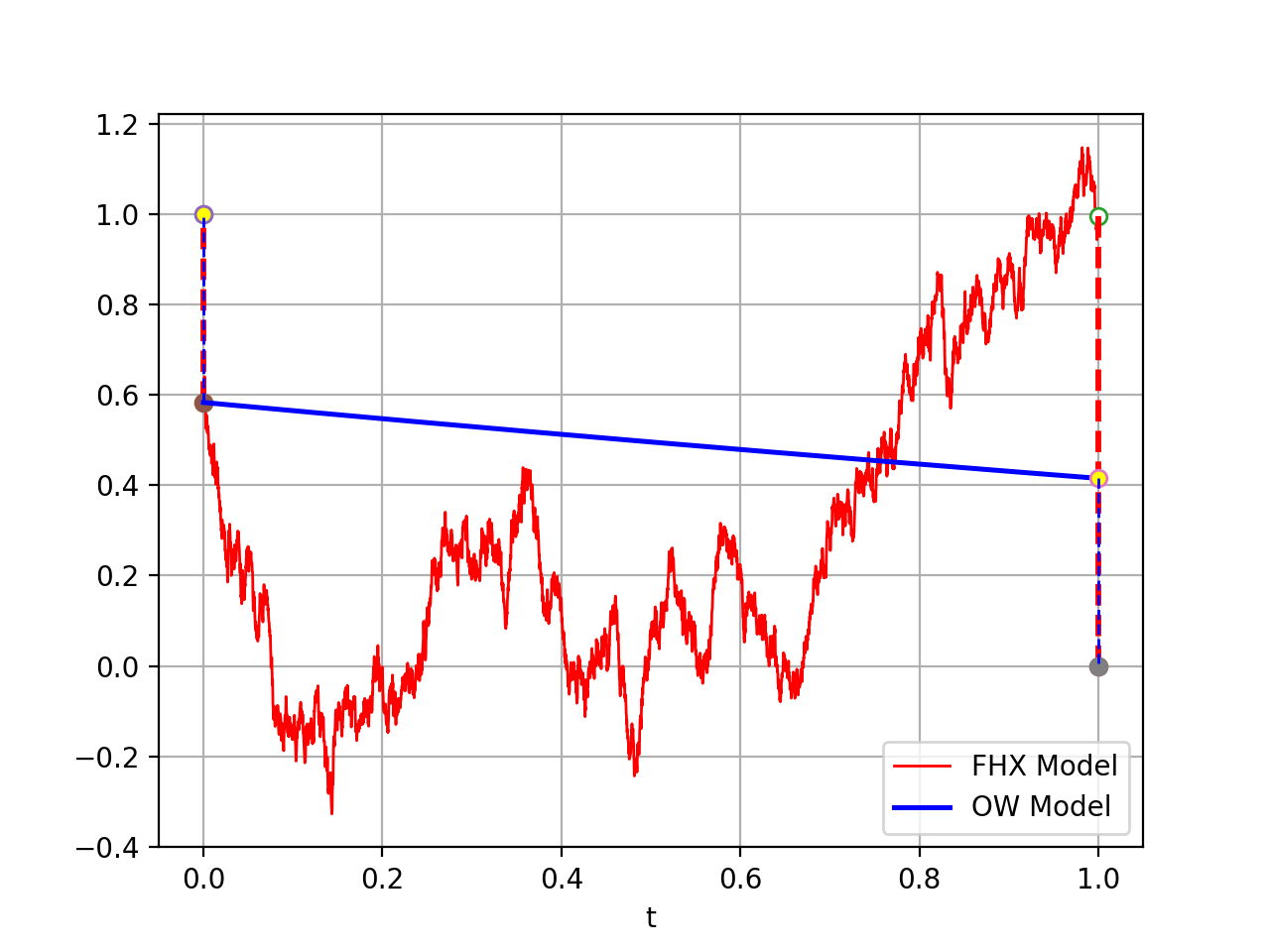}
	\end{minipage}
	\begin{minipage}[c]{0.5\textwidth}
		\centering
		\includegraphics[height=4.5cm,width=7.5cm]{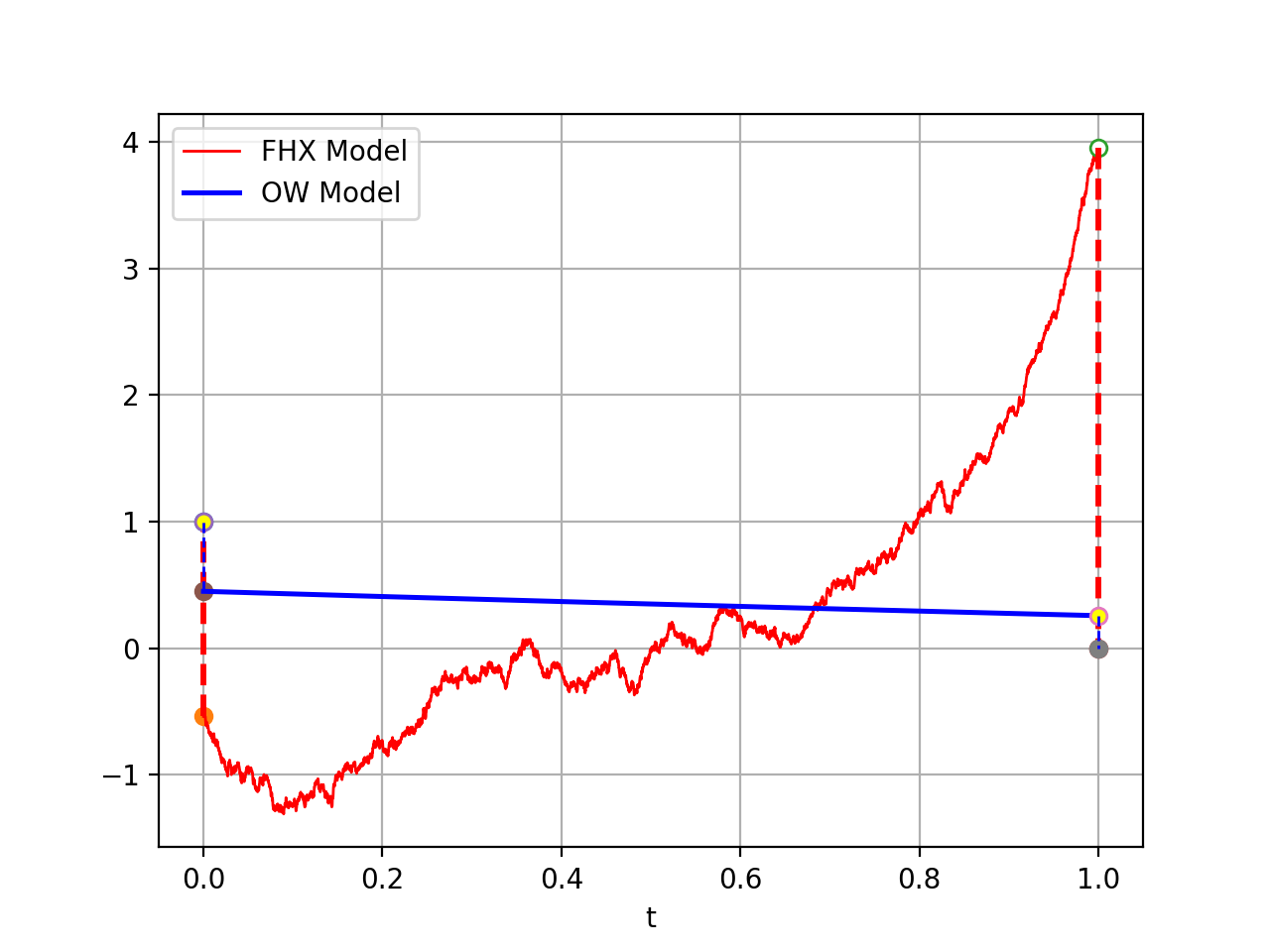}
	\end{minipage}
	\caption{Dependence of the optimal position on the impact parameter $\alpha$ for $\alpha = 0$ (left) and $\alpha = 1.8$ (right). Other parameters are chosen as $\rho=0.4$, $\lambda = 0$, $\gamma_1 = 0.1$, $\gamma_2 = 0.5$, $\beta=3$.} 
\end{figure}

Figure 2 displays the optimal position for varying degrees of child order flow. Even in the absence of any feedback effect ($\alpha=0$) we see that it may be optimal to take short positions, that is, to drive the benchmark price down and then to close the position submitting a large order at a favorable price at the end of the trading period. This effect is much stronger in the presence of child order flow where the price decrease due to own selling may be very strong and may well outweigh the cost of block trade at the end of the trading period. Optimal positions for different transient market impact parameters are shown in Figure 3.

\begin{figure}[h]\label{fig3}
	\begin{minipage}[c]{0.5\textwidth}
		\centering
		\includegraphics[height=4.5cm,width=7.5cm]{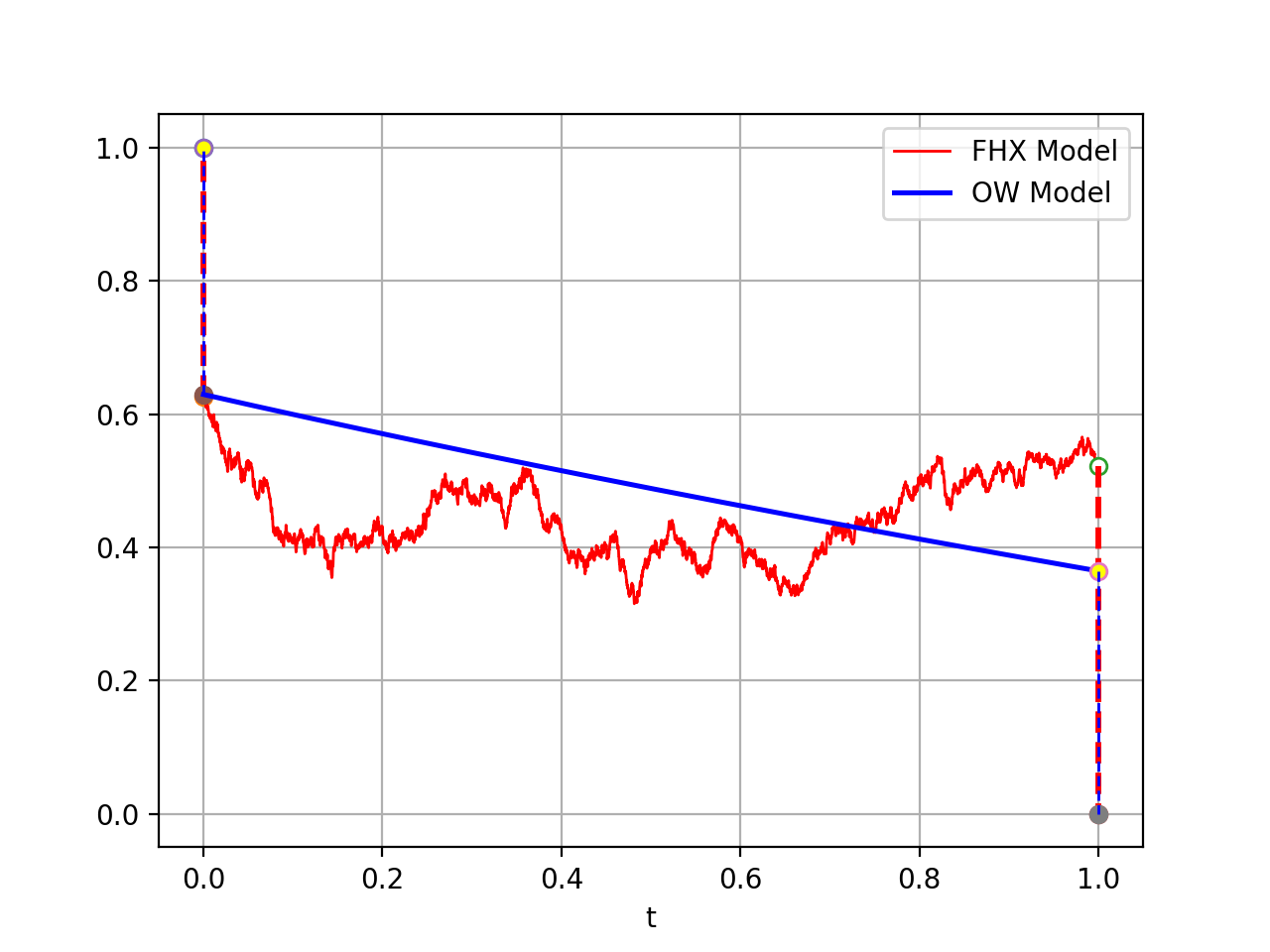}
	\end{minipage}
	\begin{minipage}[c]{0.5\textwidth}
		\centering
		\includegraphics[height=4.5cm,width=7.5cm]{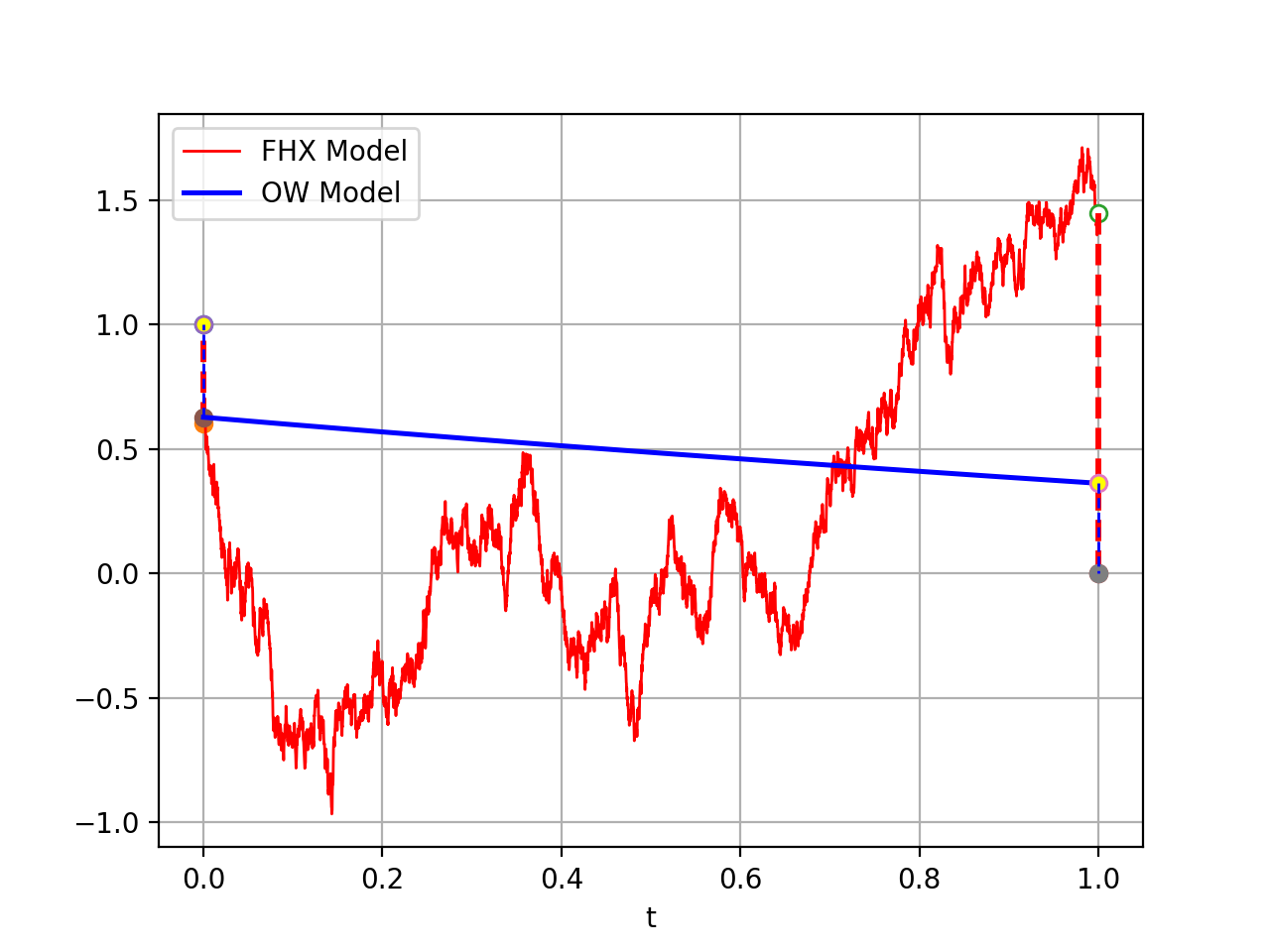}
	\end{minipage}
\caption{Dependence of the optimal position on the impact parameter $\gamma_2$ for $\gamma_2 = 2$ (left) and $\gamma_2 = 0.3$ (right). Other parameters are chosen as $\rho = 0.7$, $\lambda = 0$, $\gamma_1 = 0.1$, $\alpha = 0.5$, $\beta=1.1$.} 
\end{figure}

\section{Conclusion}\label{sec-con}

We considered a novel mean-field control problem with semimartingale strategies. We obtained a candidate value function by passing to the limit from a sequence of discrete time models. The value function can be described in terms of the solution to a fully coupled system of Riccati equations. A sophisticated transformation shows that the system has a unique solution and that the candidate optimal strategy is indeed optimal.  

Several avenues are open for future research. Let us just mention two. First, except the volatility of the spread all cost coefficients in our model are deterministic constants. Although there are many liquidation models where similar assumptions are made, these assumptions seem restrictive from a mathematical perspective. However, as far as we can tell there is no obvious way to extend the heuristics outlined in the appendix to the case of random coefficients. Second, we only considered a single-player model. While there is a substantial literature on $N$-player games and, more so on mean-field games (MFGs) with singular controls (see e.g. \cite{Campi-2020,Ferrari-2021,Cao-2020,DFFN-2022,Fu2019,FH-2017,GX-2020}), 
MFGs with semimartingale strategies have not yet been considered in the literature to the best of our knowledge.


\begin{appendix}
	
\section{Heuristic derivation of the optimal solution}
In this appendix we consider a discrete time model from which we heuristically derive the Riccati equations in terms of which we can represent both the value function and the optimal strategy in our continuous time model. The idea is to derive a recursive dynamics in a discrete-time setting and then to take formal limits as the time between two consecutive trading times tends to zero. 

\subsection{The discrete time model}
Let us assume that there are $N+1$ trading times $0,\Delta,2\Delta,\cdots,N\Delta$. Let the volume traded at time $i\Delta$ be denoted by $\xi_i$. The state immediately before this control is implemented is denoted by
\[
	\mathcal X_{i \Delta -} := (X_{i\Delta-},Y_{i\Delta-},C_{i\Delta-}). 
\]	
Let $(\epsilon_n)$ be a sequence of  i.i.d.~$N(0,\Delta)$-distributed random variables. In terms of the quantities
\[
	\mathcal L=\begin{pmatrix}0~ 1~ 0\end{pmatrix}^\top,\quad \mathcal R_{}=\frac{\gamma_{2}}{2},\quad \mathcal Q_\Delta=\begin{pmatrix}		
	\Delta\lambda_{}&0&0\\
	0&0&0\\
	0&0&0	
	\end{pmatrix}
\]
the discrete-time cost functional is given by
\[
J(\xi)=\mathbb E\left[	\sum_{i=0}^N \mathcal L^\top\mathcal X_{i\Delta-}\xi_i  +\mathcal R_{}\xi^2_i+\mathcal X_{i\Delta-}^\top\mathcal Q_{\Delta}\mathcal X_{i\Delta-}				\right],
\]
and the discrete time state dynamics in matrix form reads 
\begin{equation}\label{eq:mathcal X}
	\mathcal X_{(n+1)\Delta-}=\mathcal A_{}\mathcal X_{n\Delta-}+\overline{\mathcal A}_{ }\mathbb E[\mathcal X_{n\Delta-}]+\mathcal B_{ }\xi_n +\overline{\mathcal B}_{ }\mathbb E[\xi_n]+\mathcal C_{}+\mathcal D_{}\epsilon_{n+1},
\end{equation}
where
\begin{equation}\label{eq:A2}
	\begin{split}
		\mathcal A_{}=&~ \begin{pmatrix}		1&0&0\\
			0& 1-\Delta\rho_{}&-\Delta\gamma_{1}(\beta-\alpha)\\
			0& 0& 1-(\beta-\alpha)\Delta	
		\end{pmatrix},\quad   \overline{\mathcal A}_{}=  \begin{pmatrix}			
			0&0&0\\
			-\alpha\Delta\gamma_{1}&0&0\\
			-\alpha\Delta&0&0
		\end{pmatrix},\\
		\mathcal B_{}=&~\begin{pmatrix}-1  	& (1-\Delta\rho_{})\gamma_{2}& 0	\end{pmatrix}^\top,\qquad\qquad  \overline{\mathcal B}_{}=\begin{pmatrix} 0& \alpha\Delta\gamma_{1} & \alpha\Delta					\end{pmatrix}^\top,\\
		\mathcal C_{}=&~\begin{pmatrix}		 0&\alpha\Delta\gamma_{1}\mathbb E[x_0]& \alpha\Delta\mathbb E[x_0]		\end{pmatrix}^\top,\quad~~  \mathcal D_{}=\begin{pmatrix}	0&\sigma_{}&0				\end{pmatrix}^\top.
	\end{split}
\end{equation}
\subsubsection{The value function}

To derive a representation of the value function in discrete time we denote by $\mu$ the law of the random variable $\mathcal X:=\mathcal X_{n\Delta-}$ and set 
\[
	V_n(\mu):=\inf_{(\xi_i)_{i=n}^N}\mathbb E\left[	\sum_{i=n}^N \mathcal L^\top\mathcal X_{i\Delta-}\xi_i  +\mathcal R_{}\xi^2_i+\mathcal X_{i\Delta-}^\top\mathcal Q_{\Delta}\mathcal X_{i\Delta-}				\right].
\]

By the dynamic programming principle given in \cite[Corollary 4.1]{Cosso2019} we have that
\begin{equation}\label{DPP}
	V_n(\mu)=\inf_{\xi}\bigg\{  \mathbb E\left[  \mathcal L^\top\mathcal X_{}\xi  +\mathcal R_{}\xi^2+\mathcal X_{}^\top\mathcal Q_\Delta\mathcal X_{}\right]+  V_{n+1} \Big( \mathbb P\circ \mathcal X_{(n+1)\Delta-}^{(\xi),-1} \Big)   \bigg\},	
\end{equation}
where $\mathcal X^{(\xi)}$ denotes the state corresponding to the control $\xi$.  Let
\[
\xi=\widehat\xi+\delta, \quad\widehat{\mathcal X}=(	\widehat{ X}_{n\Delta-},\widehat Y_{n\Delta-},\widehat C_{n\Delta-}	)^\top=(	X_{n\Delta-}-\delta, Y_{n\Delta-}+\gamma_2\delta, C_{n\Delta-}		)^\top=\mathcal X+(-\delta,\gamma_2\delta,0)^\top.
\]    		
A straightforward calculation shows that
\begin{equation}\label{recursive-Vn}
	\begin{split}
		V_n(\mu)=&~ \inf_{\widehat\xi}  \Big\{	\mathbb E\left[		 \mathcal L^\top\widehat{\mathcal X}_{}\widehat\xi  +\mathcal R_{}\widehat\xi^2+\widehat{\mathcal X}_{}^\top\mathcal Q_\Delta\widehat{\mathcal X}_{}+  V_{n+1}(\mathbb P\circ \widehat{\mathcal X}_{(n+1)\Delta-}^{(\widehat\xi),-1} ) 		\right]							\Big\}\\
		&~+\mathbb E\left[	\delta Y_{n\Delta-}+\frac{\gamma_{2}}{2}\delta^2 +\lambda_{}\Delta\delta^2+2\lambda_{}\Delta\delta\widehat X_{n\Delta-}			\right]\\
		=&~V_n(\widehat\mu)+\mathbb E\left[	\delta Y_{n\Delta-}+\frac{\gamma_{2}}{2}\delta^2 +\lambda_{}\Delta\delta^2+2\lambda_{}\Delta\delta\widehat X_{n\Delta-}			\right]\\
		=&~V_n(\widehat\mu)+\mathbb E\left[		(2\lambda_{}\Delta\delta,\delta,0)\mathcal X		\right]	+\frac{\gamma_{2}}{2}\delta^2-\lambda_{}\Delta\delta^2\\
		=&~V_n(\widehat\mu)+	(2\lambda_{}\Delta\delta,\delta,0)\overline{\mu}	+\frac{\gamma_{2}}{2}\delta^2-\lambda_{}\Delta\delta^2,
	\end{split}
\end{equation}
where $\widehat{\mathcal X}$ follows the distribution $\widehat\mu$. 
Let us now make the ansatz
\begin{equation}\label{ansatz}
	V_n(\mu)=		\textnormal{Var}(\mu)(A_{n})  + \overline\mu^\top B_{n}\overline\mu		+ D_{n}^\top\overline\mu+F_{n},
\end{equation}
where for each $n$, $A_{n},B_{n}\in\mathbb S^3$, $D_{n}\in\mathbb R^3$ and $F_{n}\in\mathbb R$ are to be determined. Along with equation \eqref{recursive-Vn} the ansatz yields that
\begin{equation*}
	\begin{split}
		&~\mathbb E\left[	\mathcal X^\top A_{n}\mathcal X		\right]		-\mathbb E[\mathcal X^\top]  A_n\mathbb E[\mathcal X]  +\mathbb E[\mathcal X^\top]  B_n\mathbb E[\mathcal X]		+ D^\top_n\mathbb E[\mathcal X]+F_n\\
		=&~\mathbb E\left[	\widehat{\mathcal X}^\top A_n\widehat{\mathcal X}		\right]		-\mathbb E[\widehat{\mathcal X}^\top] A_n\mathbb E[\widehat{\mathcal X}]  +\mathbb E[\widehat{\mathcal X}^\top] B_n\mathbb E[\widehat{\mathcal X}]		+D^\top_n\mathbb E[\widehat{\mathcal X}]+F_n\\
		&~+(2\lambda_{}\Delta\delta,\delta,0)\overline\mu	+ \frac{\gamma_{2}}{2}\delta^2-\lambda_{}\Delta\delta^2.
	\end{split}
\end{equation*}

Dividing by $\delta$ on both sides and letting $\delta\rightarrow 0$, we get that
\[
	(-1,\gamma_{2},0)B_n\overline\mu  +\frac{1}{2}D_n^\top(-1,\gamma_{2},0)^\top		+(	\lambda\Delta,\frac{1}{2},0	)\overline\mu =0.
\]
The fact that this equation needs to hold for all $\bar \mu$ suggests that the coefficients multiplying the entries of the vector $\overline\mu$ are all equal to zero and so are the entries of the second summand. This yields that
\begin{equation}\label{relation-Bn-Dn}
	\left\{\begin{split}
		&~- B_{n,11}+\gamma_{2} B_{n,21}+\lambda_{}\Delta=0\\
		&~- B_{n,12}+\gamma_{2} B_{n,22}+\frac{1}{2}=0\\
		&~- B_{n,13}+\gamma_{2}  B_{n,23}=0\\
		&~- D_{n,1}+\gamma_{2} D_{n,2}=0.
	\end{split}\right.
\end{equation}
We conjecture that the entries of the matrix $A_n$ satisfy the same equations as those of $B_n$, i.e.
\begin{equation}\label{relation-An}
	\left\{\begin{split}
		&~- A_{n,11}+\gamma_{2} A_{n,21}+\lambda_{}\Delta=0\\
		&~- A_{n,12}+\gamma_{2} A_{n,22}+\frac{1}{2}=0\\
		&~- A_{n,13}+\gamma_{2}  A_{n,23}=0.
	\end{split}\right.
\end{equation}

From the equations \eqref{DPP} and \eqref{ansatz}, we will derive recursive equations for the coefficients $ A_n$, $ B_n$, $ D_n$ and $ F_n$ from which we will then derive the candidate continuous time dynamics.


\subsubsection{The optimal strategy for the discrete time model}\label{sec:app-optimal_strategy}

To derive the desired recursion, we first need to determine the candidate optimal strategy for the discrete time model. To this end, we first use \eqref{eq:mathcal X} and \eqref{eq:A2} to conclude that
\begin{equation}
	\begin{split}
		\mathcal X_{(n+1)\Delta-}-\mathbb E\left[	\mathcal X_{(n+1)\Delta-}	\right]=\mathcal A_{}\left(\mathcal X_{n\Delta-}-\mathbb E[\mathcal X_{n\Delta-}]\right)+\mathcal B_{}\left(\xi_n -\mathbb E[\xi_n]\right)+\mathcal D_{}\epsilon_{n+1}.
	\end{split}
\end{equation}
Hence,
\begin{equation*}
	\begin{split}
		& \text{Var} \Big(	\mathbb P\circ\mathcal X^{(\xi),-1}_{(n+1)\Delta-}\Big)(A_{n+1} ) \\
		= & ~ \mathbb E\left[	\left(\mathcal X^{(\xi)}_{(n+1)\Delta-}-\mathbb E\left[		\mathcal X^{(\xi)}_{(n+1)\Delta-}	\right]\right)^\top A_{n+1}\left(\mathcal X_{(n+1)\Delta-}	-	\mathbb E\left[	\mathcal X^{(\xi)}_{(n+1)\Delta-}	\right]	\right)			\right]        \\
		=&~\mathbb E\Big[	\Big\{	\mathcal A_{}\left(\mathcal X_{n\Delta-}-\mathbb E[\mathcal X_{n\Delta-}]\right)+\mathcal B_{}\left(\xi_n -\mathbb E[\xi_n]\right)+\mathcal D_{}\epsilon_{n+1}			\Big\}^\top A_{n+1} \\
		&~\cdot\Big\{	\mathcal A_{}\left(\mathcal X_{n\Delta-}-\mathbb E[\mathcal X_{n\Delta-}]\right)+\mathcal B_{}\left(\xi_n -\mathbb E[\xi_n]\right)+\mathcal D_{}\epsilon_{n+1}				\Big\}				\Big]\\
		=&~\mathbb E\Big[	\Big\{			\mathcal A_{}\left(\mathcal X_{n\Delta-}-\mathbb E[\mathcal X_{n\Delta-}]\right)+\mathcal B_{}\left(\xi_n -\mathbb E[\xi_n]\right)			\Big\}^\top A_{n+1}\Big\{		\mathcal A_{}\left(\mathcal X_{n\Delta-}-\mathbb E[\mathcal X_{n\Delta-}]\right)+\mathcal B_{}\left(\xi_n -\mathbb E[\xi_n]\right)		\Big\}			\Big]\\
		&~	+ \Delta\mathcal D^\top_{} A_{n+1}\mathcal D_{}  .
	\end{split}
\end{equation*}
Moreover, 
\begin{equation*}
	\begin{split}
		\overline{	\mathbb P\circ\mathcal X^{(\xi),-1}_{(n+1)\Delta-}  } = \mathbb E\left[		\mathcal X^{(\xi)}_{(n+1)\Delta-} 	\right]=(	\mathcal A_{}+\overline{\mathcal A}_{ }		)\overline\mu +(\mathcal B_{ }+\overline{\mathcal B}_{ })\mathbb E[\xi]  +\mathcal C_{ }.
	\end{split}
\end{equation*}

From the ansatz \eqref{ansatz} we get that
\begin{equation*}
	\begin{split}
		&~V_{n+1} \Big( \mathbb P\circ\mathcal X^{(\xi),-1}_{(n+1)\Delta-} \Big)\\
		=&~\mathbb E\Big[	\Big\{			\mathcal A_{}\left(\mathcal X_{n\Delta-}-\overline\mu\right)+\mathcal B_{ }\left(\xi -\mathbb E[\xi]\right)			\Big\}^\top A_{n+1}\Big\{		\mathcal A_{}\left(\mathcal X_{n\Delta-}-\overline\mu\right)+\mathcal B_{ }\left(\xi-\mathbb E[\xi]\right)		\Big\}			\Big]\\
		&~	+ \Delta\mathcal D^\top_{} A_{n+1}\mathcal D_{} + \Big\{(	\mathcal A_{}+\overline{\mathcal A}_{}		)\overline\mu +(\mathcal B_{}+\overline{\mathcal B}_{})\mathbb E[\xi]  +\mathcal C_{}		\Big\}^\top B_{n+1} \Big\{(	\mathcal A_{}+\overline{\mathcal A}_{}		)\overline\mu +(\mathcal B_{}+\overline{\mathcal B}_{})\mathbb E[\xi]  +\mathcal C_{}		\Big\}\\
		&~+D_{n+1}^\top \Big\{(	\mathcal A_{}+\overline{\mathcal A}_{}		)\overline\mu +(\mathcal B_{}+\overline{\mathcal B}_{})\mathbb E[\xi]  +\mathcal C_{}		\Big\}	+ F_{n+1}.	
	\end{split}
\end{equation*}
Thus, the DPP \eqref{DPP} implies that
\begin{equation*}
	\begin{split}
		V_n(\mu)
		=&~\inf_\xi 	\Bigg\{	\mathbb E\left[	 \mathcal L^\top\mathcal X_{n\Delta-}\xi  +\mathcal R_{}\xi^2+\mathcal X_{n \Delta-}^\top\mathcal Q_{\Delta}\mathcal X_{n\Delta-}		\right]		\\
		&~+\mathbb E\Big[	\Big\{			\mathcal A_{}\left(\mathcal X_{n\Delta-}-\overline\mu\right)+\mathcal B_{}\left(\xi -\mathbb E[\xi]\right)			\Big\}^\top A_{n+1}\Big\{		\mathcal A_{}\left(\mathcal X_{n\Delta-}-\overline\mu\right)+\mathcal B_{}\left(\xi-\mathbb E[\xi]\right)		\Big\}			\Big]	+ \Delta_{}\mathcal D^\top A_{n+1}\mathcal D_{}	\\
		&~	+ \Big\{(	\mathcal A_{}+\overline{\mathcal A}_{}		)\overline\mu +(\mathcal B_{}+\overline{\mathcal B}_{})\mathbb E[\xi]  +\mathcal C_{}		\Big\}^\top B_{n+1} \Big\{(	\mathcal A_{}+\overline{\mathcal A}_{}		)\overline\mu +(\mathcal B_{}+\overline{\mathcal B}_{})\mathbb E[\xi]  +\mathcal C_{}		\Big\}\\
		&~+D_{n+1}^\top \Big\{(	\mathcal A_{ }+\overline{\mathcal A}_{}		)\overline\mu +(\mathcal B_{}+\overline{\mathcal B}_{})\mathbb E[\xi]  +\mathcal C_{}		\Big\}	+ F_{n+1}\Bigg\}\\
		:= &~ \inf_\xi \mathcal J(\xi).
	\end{split}
\end{equation*}
Let $\xi^*$ be the candidate optimal strategy and let $\widetilde\xi$ be an arbitrary strategy. Then
\begin{equation*}
	\begin{split}
		\mathcal J(\xi^*+\varepsilon\widetilde\xi)-\mathcal J(\xi^*)
		=&~\varepsilon\mathbb E\bigg\{		\mathcal L^\top \mathcal X_{n\Delta-}\widetilde\xi 
		+2\mathcal R_{}\xi^*\widetilde\xi	\\
		&~\quad		+2\Big( \mathcal A_{}(\mathcal X_{n\Delta-}-\overline\mu) +\mathcal B_{}(\xi^*-\mathbb E[\xi^*])		\Big)^\top A_{n+1}\mathcal B_{}( \widetilde\xi-\mathbb E[\widetilde\xi] )							\\
		&~\quad+ 2\Big(	(\mathcal A_{}+\overline{\mathcal A}_{})\overline\mu+(\mathcal B_{}+\overline{\mathcal B}_{})\mathbb E[\xi^*]	+\mathcal C_{}			\Big)^\top B_{n+1}(	\mathcal B_{}+\overline{\mathcal B}_{}		)\mathbb E[\widetilde\xi]\\
		&~\quad+D_{n+1}^\top(\mathcal B_{}+\overline{\mathcal B}_{})\mathbb E[\widetilde\xi] \bigg\}+O(\varepsilon^2)\\ 
		=&~\varepsilon \mathbb E\bigg\{	 \Big[	\mathcal L^\top \mathcal X_{n\Delta-}  +2\mathcal R_{}\xi^* +2\Big( \mathcal A_{}(\mathcal X_{n\Delta-}-\overline\mu) +\mathcal B_{}(\xi^*-\mathbb E[\xi^*])	\Big)^\top A_{n+1}\mathcal B_{} \\
		&~\quad +2\Big(	( \mathcal A_{}+\overline{\mathcal A}_{} )\overline\mu+(\mathcal B_{}+\overline{\mathcal B}_{})\mathbb E[\xi^*] +\mathcal C_{}	\Big)^\top B_{n+1}( \mathcal B_{}+\overline{\mathcal B}_{} ) +D_{n+1}(\mathcal B_{}+\overline{\mathcal B}_{})		\Big]	\widetilde\xi			\bigg\}\\
		&~\quad+O(\varepsilon^2).
	\end{split}
\end{equation*}
Since $\liminf_{\varepsilon\rightarrow 0} \frac{ \mathcal J(\xi^*+\varepsilon\widetilde\xi) -\mathcal J(\xi^*) }{\varepsilon}\geq 0$ for any $\widetilde\xi$, we conclude that
\begin{equation}\label{eq:xi*}
	\begin{split}
		&~\mathcal L^\top \mathcal X_{n\Delta-}  +2\mathcal R_{}\xi^* +2\Big( \mathcal A_{}(\mathcal X_{n\Delta-}-\overline\mu) +\mathcal B_{}(\xi^*-\mathbb E[\xi^*])	\Big)^\top A_{n+1}\mathcal B_{} \\
		&~+2\Big(	( \mathcal A_{}+\overline{\mathcal A}_{} )\overline\mu+(\mathcal B_{}+\overline{\mathcal B}_{})\mathbb E[\xi^*] +\mathcal C_{}	\Big)^\top B_{n+1}( \mathcal B_{}+\overline{\mathcal B}_{} )\\
		&~+D_{n+1}^\top(	\mathcal B_{}+\overline{\mathcal B}_{}		) \\
		= & ~ 0.
	\end{split}
\end{equation}
Taking expectations on both sides of \eqref{eq:xi*}, we get that
\begin{equation}\label{eq:E[xi*]}
	\begin{split}
		\mathbb E[\xi^*]
		= &~-\bigg\{	2\mathcal R_{}  +2(\mathcal B_{}+\overline{\mathcal B}_{}	)^\top B_{n+1}(	\mathcal B_{}+\overline{\mathcal B}_{}		)			\bigg\}^{-1} \bigg\{		\mathcal L^\top + 2(	\mathcal B_{}+\overline{\mathcal B}_{}		)^\top B_{n+1}(		\mathcal A_{}+\overline{\mathcal A}_{}	)	\bigg\}\overline\mu\\
		&~-\bigg\{	2\mathcal R_{}  +2(\mathcal B_{}+\overline{\mathcal B}_{}	)^\top B_{n+1}(	\mathcal B_{}+\overline{\mathcal B}_{}		)			\bigg\}^{-1} \bigg\{	2\mathcal C_{}^\top B_{n+1}(	\mathcal B_{}+\overline{\mathcal B}_{}	)		+D_{n+1}^\top(\mathcal B_{}+\overline{\mathcal B}_{})		\bigg\}.
	\end{split}
\end{equation}
Combining \eqref{eq:xi*} and \eqref{eq:E[xi*]} yields that
\begin{equation}\label{eq:xi*-E[xi*]}
	\begin{split}
		\xi^*-\mathbb E[\xi^*]= -\Big\{		2\mathcal R_{}	 +2\mathcal B_{}^\top A_{n+1}\mathcal B_{}	\Big\}^{-1}  \left(	\mathcal L^\top+2\mathcal B_{}^\top A_{n+1}\mathcal A_{}		\right)(\mathcal X_{n\Delta-}-\overline\mu),
	\end{split}
\end{equation}
and from \eqref{eq:E[xi*]} and \eqref{eq:xi*-E[xi*]} we have that
\begin{equation*}\label{eq:xi*-2}
	\begin{split}
		\xi^*=&~-\Big\{		2\mathcal R_{}	 +2\mathcal B_{}^\top A_{n+1}\mathcal B_{}	\Big\}^{-1}  \left(	\mathcal L^\top+2\mathcal B_{}^\top A_{n+1}\mathcal A_{}		\right)(\mathcal X_{n\Delta-}-\overline\mu)\\
		&~-\bigg\{	2\mathcal R_{}  +2(\mathcal B_{}+\overline{\mathcal B}_{}	)^\top B_{n+1}(	\mathcal B_{}+\overline{\mathcal B}_{}		)			\bigg\}^{-1} \bigg\{		\mathcal L^\top + 2(	\mathcal B_{}+\overline{\mathcal B}_{}		)^\top B_{n+1}(		\mathcal A_{}+\overline{\mathcal A}_{}	)	\bigg\}\overline\mu\\
		&~-\bigg\{	2\mathcal R_{}  +2(\mathcal B_{}+\overline{\mathcal B}_{}	)^\top B_{n+1}(	\mathcal B_{}+\overline{\mathcal B}_{}		)			\bigg\}^{-1} \bigg\{	2\mathcal C_{}^\top B_{n+1}(	\mathcal B_{}+\overline{\mathcal B}_{}	)		+D_{n+1}^\top(\mathcal B_{}+\overline{\mathcal B}_{})		\bigg\}.
	\end{split}
\end{equation*}
In terms of the notation
\begin{equation}\label{eq:IABDn}
	\begin{split}
		\tilde a_n=&~	\mathcal R_{}	 +\mathcal B_{}^\top A_{n+1}\mathcal B_{},	  \qquad \qquad \quad
		a_n=~\mathcal R_{}  +(\mathcal B_{}+\overline{\mathcal B}_{}	)^\top B_{n+1}(	\mathcal B_{}+\overline{\mathcal B}_{}		),			\\
		I^{A}_n=&\frac{1}{2}\mathcal L^\top+\mathcal B_{}^\top A_{n+1}\mathcal A_{}, \qquad \qquad 	
		I^{B}_{n} =	\frac{1}{2}\mathcal L^\top +(	\mathcal B_{}+\overline{\mathcal B}_{}		)^\top B_{n+1}(		\mathcal A_{}+\overline{\mathcal A}_{}	),\\
		I^D_{n}=&~	\mathcal C_{}^\top B_{n+1}(	\mathcal B_{}+\overline{\mathcal B}_{}	)		+\frac{1}{2}D_{n+1}^\top(\mathcal B_{}+\overline{\mathcal B}_{}),
	\end{split}
\end{equation}
the optimal strategy at time $n\Delta$ as a function of the initial state $\mathcal X_{n\Delta-}$ can be written as
\begin{equation}\label{eq:xi*n}
	\begin{split}
	\xi^*_n=\xi^*_n(\mathcal X_{n\Delta-})=&~-\frac{I_{n}^A}{\tilde a_n}(\mathcal X_{n\Delta-}-\overline\mu)-\frac{I_{n}^B}{a_n}\overline\mu-\frac{I_{n}^D}{a_n}.
	\end{split}
\end{equation}



\subsection {Heuristic derivation of $(A_t,B_t,D_t,F_t)$}\label{sec:heuristic-Riccati}

To obtain the continuous time dynamics of the coefficient processes $(A,B,D,F)$ we are now going to derive a recursive dynamics for the processes $A_n, B_n, D_n$ and $F_n$ and then formal derivatives. 

Taking the equation \eqref{eq:xi*n} back into the cost function $\mathcal J(\cdot)$ we get that
\begin{align*}
		\mathcal J(\xi^*_n)=&~		\mathbb E\Big\{		\mathcal L^\top(\mathcal X_{n\Delta-}-\overline\mu)\xi^*_n +\mathcal L^\top\overline\mu\xi^*_n +\mathcal R_{}(\xi^*_n)^2+(\mathcal X_{n\Delta-}-\overline\mu)^\top\mathcal Q_{\Delta}(\mathcal X_{n\Delta-}-\overline\mu) +\overline\mu^\top\mathcal Q_{\Delta}\overline\mu	\Big\}\\
		&~+\mathbb E\Big\{	(		\mathcal X_{n\Delta-}-\overline\mu		)^\top\mathcal A_{}^\top A_{n+1}\mathcal A_{}	(	\mathcal X_{n\Delta-}-\overline\mu		)	+ 2(	\mathcal X_{n\Delta-}-\overline\mu		)^\top\mathcal A_{}^\top A_{n+1}\mathcal B_{} (\xi^*_n-\mathbb E[\xi^*_n])		\\
		&~+\mathcal B_{}^\top A_{n+1}\mathcal B_{}(\xi^*_n-\mathbb E[\xi^*_n])^2	+\Delta\mathcal D^\top_{} A_{n+1}\mathcal D_{}			\Big\}\\
		&~+\overline\mu^\top(	\mathcal A_{}+\overline{\mathcal A}_{}			)^\top B_{n+1}(	\mathcal A_{}+\overline{\mathcal A}_{}		)\overline\mu	+ 2\overline\mu^\top (\mathcal A_{}+\overline{\mathcal A}_{})^\top B_{n+1}(\mathcal B_{}+\overline{\mathcal B}_{})\mathbb E[\xi^*_n]\\
		&~+2( \mathcal B_{}+\overline{\mathcal B}_{}			)^\top B_{n+1}\mathcal C_{}\mathbb E[\xi^*_n]+(\mathcal B_{}+\overline{\mathcal B}_{})^\top B_{n+1}(	\mathcal B_{}+\overline{\mathcal B}_{}	)\mathbb E[\xi^*_n]^2\\
		&~+2\overline\mu^\top (		\mathcal A_{}+\overline{\mathcal A}_{}		)^\top B_{n+1}\mathcal C_{}	+ \mathcal C^\top_{} B_{n+1}\mathcal C_{}\\
		&~+D^\top_{n+1}( \mathcal A_{}+\overline{\mathcal A}_{}		)\overline\mu	+ D_{n+1}^\top(		\mathcal B_{}+\overline{\mathcal B}_{}	)\mathbb E[\xi^*_n]  + D_{n+1}^\top\mathcal C_{} + F_{n+1}\\
		=&~\text{Var}(\mu)\Big(	\mathcal Q_\Delta+\mathcal A^\top A_{n+1}\mathcal A-(I^A_n)^\top \widetilde a_n^{-1} I^A_n			\Big)+\overline\mu^\top \Big(	\mathcal Q_\Delta+(\mathcal A+\overline{\mathcal A})^\top B_{n+1}( \mathcal A+\overline{\mathcal A} ) -(I_n^B)^\top a_n^{-1} I_n^B			\Big)\overline\mu \\
		&~+\Big( -2a_n^{-1} I_n^D I_n^B + ( 2\mathcal C^\top B_{n+1}+D_{n+1}^\top  )(\mathcal A+\overline{\mathcal A})			\Big)\overline\mu\\
		&~-\frac{(I_n^D)^2}{a_n}+\Delta\mathcal D^\top A_{n+1}\mathcal D +\mathcal C^\top B_{n+1}\mathcal C + D^\top_{n+1}\mathcal C+F_{n+1}.
\end{align*}
Comparing the coefficients of $\text{Var}(\mu)(\cdots)$, $\overline\mu^\top (\cdots)\overline\mu$, $(\cdots)\overline\mu$ and the remaining terms respectively, we see that
 $A_n$, $B_n$, $D_n$ and $F_n$ satisfy the following recursive equations:
\begin{equation}\label{eq:An}
	\left\{\begin{split}
		A_n=&~\mathcal Q_{\Delta}+\mathcal A_{}^\top A_{n+1}\mathcal A_{}-(I^A_n)^\top\tilde a_n^{-1} I_n^A,	\\	
		B_n=&~\mathcal Q_{\Delta}+(\mathcal A+\overline{\mathcal A}  )^\top B_{n+1}( \mathcal A+\overline{\mathcal A}  ) -( I_n^B  )^\top a_n^{-1} I_n^B,\\
		D_n^\top=&~-2 a_n^{-1} I_n^D I_n^B + (  2\mathcal C^\top B_{n+1}+D^\top_{n+1 } )(\mathcal A+\overline{\mathcal A}),\\
		F_n=&~-\frac{( I_n^D )^2}{4a_n} + \Delta \mathcal D^\top A_{n+1}\mathcal D + \mathcal C^\top B_{n+1} \mathcal C + D_{n+1}^\top\mathcal C+F_{n+1}.
	\end{split}\right.
\end{equation}

Let $(A,B,D,F)$ be the formal limit of $(A_n,B_n,D_n,F_n)$ as $\Delta \to 0$. Using \eqref{relation-Bn-Dn} and \eqref{relation-An}, the driver of $(A,B,D,F)$ is formally obtained by taking the limit
$$
\lim_{\Delta\rightarrow 0} 	\frac{\Phi_{n+1}-\Phi_n}{\Delta},\qquad \Phi=A,B,D,F.
$$
In terms of the notation
\begin{equation*}
	\begin{split}
		J^A=&~\gamma_1A_{11}+\gamma_2A_{13},\quad \widetilde J^A=\gamma_1 A_{13}+\gamma_2 A_{33},\quad  \widehat J^A=-\rho A_{13}+\frac{\gamma_1(\beta-\alpha)}{2},\quad \\ J^B=&~\gamma_1B_{11}+\gamma_2B_{13},\quad   \widetilde J^B=\gamma_1 B_{13}+\gamma_2 B_{33},\quad \widehat J^B=-\rho B_{13}+\frac{\gamma_1(\beta-\alpha)}{2},\quad \breve J^B=-\rho B_{11}+\frac{\alpha\gamma_1}{2},\\
		J^D=&~\gamma_1 D_1+\gamma_2 D_3,\quad \breve J^D=-\rho D_1-\alpha\gamma_1\mathbb E[x_0],
	\end{split}
\end{equation*}
we obtain
\begin{itemize}
\item the following matrix-valued ODE for $A$: 
\begin{small}
\begin{equation}\label{eq:A}
	\left\{\begin{split}
		-\frac{dA_t}{dt}=&~\begin{pmatrix}
			0& -\rho \frac{A_{11,t}}{\gamma_2} & -\frac{\beta-\alpha}{\gamma_2} J^A_t\\
			-\rho \frac{A_{11,t}}{\gamma_2}&-\rho \frac{2A_{11,t}-\gamma_2}{\gamma_2^2}&  -\frac{\beta-\alpha}{\gamma_2^2}J^A_t+\frac{\widehat J^A_t}{\gamma_2}\\
			-\frac{\beta-\alpha}{\gamma_2} J^A_t&-\frac{\beta-\alpha}{\gamma_2^2}J^A_t+\frac{\widehat J^A_t}{\gamma_2}&-\frac{2(\beta-\alpha)}{\gamma_2}\widetilde J^A_t	\end{pmatrix}\\
		&-\frac{1}{\lambda+\gamma_2\rho}\begin{pmatrix}
			-\rho A_{11,t}-\lambda\\
			-\rho \frac{A_{11,t}}{\gamma_2}+\rho\\
			\frac{\gamma_1(\beta-\alpha)}{2}-\rho A_{13,t}
		\end{pmatrix}\begin{pmatrix}
		-\rho A_{11,t}-\lambda\\
		-\rho \frac{A_{11,t}}{\gamma_2}+\rho\\
		\frac{\gamma_1(\beta-\alpha)}{2}-\rho A_{13,t}
	\end{pmatrix}^\top\\
	&~+\begin{pmatrix}
			\lambda& 0&0\\
			0&0&0\\
			0&0&0
		\end{pmatrix}\\
		A_T=&\begin{pmatrix}
			\frac{\gamma_2}{2}& \frac{1}{2}&0\\
			\frac{1}{2}&0&0\\
			0&0&0
		\end{pmatrix},
	\end{split}\right.
\end{equation}
\end{small}
\item the following matrix-valued Riccati equation for $B$:
\begin{small}
\begin{equation}\label{eq:B}
	\left\{\begin{split}
		&-\frac{dB_t}{dt}=
		\begin{pmatrix}
			-\frac{2\alpha}{\gamma_2}J^B_t& 	-\frac{\alpha}{\gamma_2^2}J^B_t+\frac{1}{\gamma_2}\breve J^B_t & -\frac{\beta-\alpha}{\gamma_2}J^B_t-\frac{\alpha}{\gamma_2}\widetilde J^B_t \\
		-\frac{\alpha}{\gamma_2^2}J^B_t+\frac{1}{\gamma_2}\breve J^B_t&-\rho\frac{2B_{11,t}-\gamma_2}{\gamma_2^2}&  -\frac{(\beta-\alpha)}{\gamma_2^2}J^B_t+\frac{1}{\gamma_2}\widehat J^B_t \\
		-\frac{\beta-\alpha}{\gamma_2}J^B_t-\frac{\alpha}{\gamma_2}\widetilde J^B_t&-\frac{(\beta-\alpha)}{\gamma_2^2}J^B_t+\frac{1}{\gamma_2}\widehat J^B_t&-\frac{2(\beta-\alpha)}{\gamma_2}\widetilde J^B_t	\end{pmatrix}\\
		&\qquad-\frac{1}{\lambda+\gamma_2\rho-\alpha\gamma_1}\begin{pmatrix}
			\frac{\alpha}{\gamma_2}J^B_t-\breve J^B_t-\lambda\\
			\frac{\alpha}{\gamma_2^2}J^B_t+\frac{1}{\gamma_2}\breve J^B_t+\frac{\rho\gamma_2-\alpha\gamma_1}{\gamma_2}\\
			\widehat J^B_t+\frac{\alpha}{\gamma_2}\widetilde J^B_t
		\end{pmatrix}\begin{pmatrix}
\frac{\alpha}{\gamma_2}J^B_t-\breve J^B_t-\lambda\\
			\frac{\alpha}{\gamma_2^2}J^B_t+\frac{1}{\gamma_2}\breve J^B_t+\frac{\rho\gamma_2-\alpha\gamma_1}{\gamma_2}\\
			\widehat J^B_t+\frac{\alpha}{\gamma_2}\widetilde J^B_t
		\end{pmatrix}^{\top}\\
		&\qquad+\begin{pmatrix}
			\lambda& 0&0\\
			0&0&0\\
			0&0&0
		\end{pmatrix}\\
		&B_T=\begin{pmatrix}
			\frac{\gamma_2}{2}& \frac{1}{2}&0\\
			\frac{1}{2}&0&0\\
			0&0&0
		\end{pmatrix},
	\end{split}\right.
\end{equation}
\end{small}
\item the following vector-valued ODE for $D$:
	\begin{equation}\label{eq:D}
		\left\{\begin{split}
			-\frac{dD^{}_t}{dt}=&~\begin{pmatrix}
				\frac{2\alpha\mathbb E[x_0]}{\gamma_2}J^B_t-\frac{\alpha}{\gamma_2} J^D_t\\
				\frac{2\alpha\mathbb E[x_0]}{\gamma_2^2}J^B_t+\frac{1}{\gamma_2}\breve J^D_t\\
				\frac{2\alpha\mathbb E[x_0]}{\gamma_2}\widetilde J^B_t-\frac{\beta-\alpha}{\gamma_2}J^D_t	\end{pmatrix}\\
			&-\frac{1}{\lambda+\gamma_2\rho-\alpha\gamma_1}\left(	\frac{\alpha}{\gamma_2}J^D_t+\breve J^D_t\right)\begin{pmatrix}
				\frac{\alpha}{\gamma_2}J^B_t+\breve J^B_t-\lambda\\
				\frac{\alpha}{\gamma_2^2}J^B_t+\frac{1}{\gamma_2}\breve J^B_t+\rho\\
				\frac{\alpha}{\gamma_2}\widetilde J^B_t+\breve J^B_t
			\end{pmatrix}\\
			D^{}_T=&\begin{pmatrix}
				0\\
				0\\
				0
			\end{pmatrix},
		\end{split}\right.
	\end{equation}
\item and the following BSDE for $F$: 
	\begin{equation}\label{eq:F}
	\left\{\begin{split}
		-d F_t=&~\bigg\{\sigma^2_t \frac{2A_{11,t}-\gamma_2}{2\gamma_2^2}+\alpha\gamma_1\mathbb E[x_0]\frac{D_{1,t}}{\gamma_2}+\alpha\mathbb E[x_0]D_{3,t}\\
		&-\frac{1}{4(\lambda+\gamma_2\rho-\alpha\gamma_1)}\left(-\alpha\gamma_1\mathbb E[x_0]+(\gamma_1\alpha-\gamma_2\rho)\frac{D_{1,t}}{\gamma_2}+\alpha D_{3,t}\right)^2\bigg\}\,dt-Z^F_t\,dW_t\\
		F_T=&0.
	\end{split}\right.
\end{equation}
\end{itemize}
In view of \eqref{relation-Bn-Dn} and \eqref{relation-An} the above systems reduce to \eqref{eq:tilde-A}-\eqref{eqn:F}.


\subsection{Heuristic derivation of the optimal strategy in continuous time}

In this section, we construct a continuous-time candidate optimal strategy by taking limits of the discrete time model. Intuitively, and in view of the results established in \cite{Hkiv,HN-2014} we expect the optimal strategy to jump only at the initial and the terminal time, and to follow an SDE on the open interval $(0,T)$.   


\subsubsection{The jumps}

The final jump size is $X_{T-}$ in order to close the open position at $T$. To determine the initial jump we first deduce from \eqref{eq:IABDn} that \begin{equation*}
	\begin{split}
		\tilde a_n=&~\Delta(\lambda+\gamma_2\rho) + O(\Delta^2),\\
		a_n=&~\Delta(	\lambda+\rho\gamma_2-\alpha\gamma_1	) + O(\Delta^2),\\
		I^A_n=&~\Delta \left( -(\rho A_{11,n\Delta}+\lambda)+O(\Delta), ( \rho-\frac{A_{11,n\Delta}}{\gamma_2}  )+O(\Delta), (	\frac{\gamma_1(\beta-\alpha)}{2} -\rho A_{13,n\Delta}			)+O(\Delta)  \right),\\
		I_n^B=&~\Delta\left( \begin{matrix}( \frac{\alpha\gamma_1}{\gamma_2}-\rho  )B_{11,n\Delta} +\alpha B_{13,n\Delta}-\lambda+\frac{\alpha\gamma_1}{2}+O(\Delta)\\		\frac{\alpha\gamma_1-\gamma_2\rho}{\gamma_2^2}B_{11,n\Delta}+\alpha \frac{B_{13,n\Delta}}{\gamma_2}+\rho-\frac{\alpha\gamma_1}{2\gamma_2}	+O(\Delta)\\
		\frac{\gamma_1(\beta-\alpha)}{2}+(\gamma_1\alpha-\gamma_2\rho)\frac{B_{13,n\Delta}}{\gamma_2}+\alpha B_{33,n\Delta}	+ O(\Delta)
		\end{matrix}			\right)^\top,\\
		I_n^D=&~-\Delta\left( \frac{\alpha\gamma_1}{2}\mathbb E[x_0]  + \frac{(\alpha\gamma_1-\rho\gamma_2)}{2\gamma_2} D_{1,n\Delta} + \frac{\alpha\Delta D_{3,n\Delta}}{2} + O(\Delta) \right).
	\end{split}
\end{equation*}
Thus, letting $n\Delta=t$ and $n\rightarrow\infty$ in \eqref{eq:xi*n}, we see  that
\begin{equation}\label{limit:xi*}
	\begin{split}
		\xi^*_n	=&~-\frac{I_{n}^A}{\tilde a_n}(\mathcal X_{n\Delta-}-\overline\mu)-\frac{I_{n}^B}{a_n}\overline\mu-\frac{I_{n}^D}{a_n} \\ 
		\rightarrow & ~ -\frac{I_{t}^A}{\tilde a}(\mathcal X_{t-}-\overline\mu)-\frac{I_{t}^B}{a}\overline\mu-\frac{I_{t}^D}{a}:=\Delta Z_t,
	\end{split}
\end{equation}
where 
\begin{equation}
	\begin{split}
		\tilde a=&\gamma_2\rho+\lambda, \qquad \qquad \qquad \qquad 
		a=\gamma_2\rho-\gamma_1\alpha+\lambda,\\
		I^A=&\begin{pmatrix}
			-\rho A_{11}-\lambda\\
			-\rho \frac{A_{11}}{\gamma_2}+\rho\\
			\frac{\gamma_1(\beta-\alpha)}{2}-\rho A_{13}
		\end{pmatrix}^\top, \quad
		I^B=\begin{pmatrix}
			\frac{\alpha\gamma_1-\gamma_2\rho}{\gamma_2}B_{11}+\alpha B_{13}-\lambda+\frac{\alpha\gamma_1}{2}\\
			\frac{\alpha\gamma_1-\gamma_2\rho}{\gamma_2^2}B_{11}+\alpha \frac{B_{13}}{\gamma_2}+\rho-\frac{\alpha\gamma_1}{2\gamma_2}\\
			\frac{\gamma_1(\beta-\alpha)}{2}+(\gamma_1\alpha-\gamma_2\rho)\frac{B_{13}}{\gamma_2}+\alpha B_{33}
		\end{pmatrix}^\top,\\
		I^D=&	-\frac{\alpha\gamma_1}{2}\mathbb E[x_0]+(\gamma_1\alpha-\gamma_2\rho)\frac{D_1}{2\gamma_2}+\frac{\alpha}{2} D_3.
	\end{split}
\end{equation}
The limit $\Delta Z_t$ in \eqref{limit:xi*} is the candidate initial jump when starting with a position ${\cal X}_{t-}$ at time $t \in [0,T)$. 


\subsubsection{The candidate strategy on $(t,T)$}

Next, we derive recursive dynamics for the discrete-time optimal strategy from which we deduce a candidate optimal continuous-time strategy. The strategy at time $(n+1)\Delta$ satisfies
\begin{align*}
		&~\xi^*_{n+1}\Big(\mathcal X_{(n+1)\Delta-}\Big)\\
		=&~-\frac{I^A_{n+1}}{\tilde a_{n+1}}\left(\mathcal X_{(n+1)\Delta-}-\mathbb E[\mathcal X_{(n+1)\Delta-}]\right)-\frac{I^B_{n+1}}{a_{n+1}}\mathbb E[\mathcal X_{(n+1)\Delta-}]-\frac{I^D_{n+1}}{a_{n+1}}\\
		=&~-\frac{I^A_{n+1}}{\tilde a_{n+1}}\Big\{		\mathcal A_{}\left(\mathcal X_{n\Delta-}-\mathbb E[\mathcal X_{n\Delta-}]\right)+\mathcal B_{}\left(\xi^*_n -\mathbb E[\xi^*_n]\right)+\mathcal D_{}\epsilon_{n+1}\Big\}\\
		&~-\frac{I^B_{n+1}}{a_{n+1}}\Big\{
		(\mathcal A_{}+\overline{\mathcal A}_{}		)\mathbb E[\mathcal X_{n\Delta-}] +(\mathcal B_{}+\overline{\mathcal B}_{})\mathbb E[\xi^*_n]  +\mathcal C_{}\Big\}-\frac{I^D_{n+1}}{a_{n+1}}\\	
		=&~-\frac{I^A_{n+1}}{\tilde a_{n+1}}\left(\mathcal X_{n\Delta-}-\mathbb E[\mathcal X_{n\Delta-}]\right)-\frac{I^B_{n+1}}{a_{n+1}}\mathbb E[\mathcal X_{n\Delta-}]-\frac{I^D_{n+1}}{a_{n+1}}\\
		&~-\frac{I^A_{n+1}}{\tilde a_{n+1}}\mathcal K\left(\xi^*_n-\mathbb E[\xi^*_n]\right)-\frac{I^B_{n+1}}{a_{n+1}}\mathcal K\mathbb E[\xi^*_n]-\frac{I^A_{n+1}}{\tilde a_{n+1}}\mathcal D_{}\epsilon_{n+1}-\frac{I^B_{n+1}}{a_{n+1}}\mathcal C_{}\\
		&~-\frac{I^A_{n+1}}{\tilde a_{n+1}}(\mathcal A_{}-\textit{I}_{3\times3})\left(\mathcal X_{n\Delta-}-\mathbb E[\mathcal X_{n\Delta-}]\right)-\frac{I^B_{n+1}}{a_{n+1}}(\mathcal A_{}+\overline{\mathcal A}_{}-\textit{I}_{3\times3})\mathbb E[\mathcal X_{n\Delta-}]\\
		&~-\frac{I^A_{n+1}}{\tilde a_{n+1}}(\mathcal B_{}-\mathcal K)\left(\xi^*_n-\mathbb E[\xi^*_n]\right)-\frac{I^B_{n+1}}{a_{n+1}}(\mathcal B_{}+\overline{\mathcal B}_{}-\mathcal K)\mathbb E[\xi^*_n],
	\end{align*}
where $I_{3\times 3}$ is the identity matrix of order $3$. 
In view of \eqref{eq:xi*n} the first line in the third equality equals $\xi^*_{n+1}(\mathcal X_{n\Delta-})$. Thus, 
\begin{align*}
&~	\xi^*_{n+1}(\mathcal X_{(n+1)\Delta-})\\
		=&~\xi^*_{n+1}(\mathcal X_{n\Delta-})-\frac{I^A_{n+1}}{\tilde a_{n+1}}\mathcal K\left(\xi^*_n-\mathbb E[\xi^*_n]\right)-\frac{I^B_{n+1}}{a_{n+1}}\mathcal K\mathbb E[\xi^*_n]-\frac{I^A_{n+1}}{\tilde a_{n+1}}\mathcal D\epsilon_{n+1}-\frac{I^B_{n+1}}{a_{n+1}}\mathcal C_{}\\
		&~-\frac{I^A_{n+1}}{\tilde a_{n+1}}(\mathcal A_{}-\textit{I}_{3\times3})\left(\mathcal X_{n\Delta-}+\mathcal K\xi^*_n-\mathbb E[\mathcal X_{n\Delta-}+\mathcal K\xi^*_n]\right)-\frac{I^B_{n+1}}{a_{n+1}}(\mathcal A_{}+\overline{\mathcal A}_{}-\textit{I}_{3\times3})\mathbb E[\mathcal X_{n\Delta-}+\mathcal K\xi^*_n]\\
		&~-\frac{I^A_{n+1}}{\tilde a_{n+1}}(\mathcal B_{}-\mathcal K-(\mathcal A_{}-\textit{I}_{3\times3})\mathcal K)\left(\xi^*_n-\mathbb E[\xi^*_n]\right)\\
		&~-\frac{I^B_{n+1}}{a_{n+1}}(\mathcal B_{}+\overline{\mathcal B}_{}-\mathcal K-(\mathcal A_{}+\overline{\mathcal A}_{}-\textit{I}_{3\times3})\mathcal K)\mathbb E[\xi^*_n].
\end{align*}
Direct computations show that
\begin{equation}
\left\{	\begin{split}
	    &I^A_{n+1}\mathcal K=\tilde a_{n+1},\quad I^B_{n+1}\mathcal K=a_{n+1},\\
		&I^A_{n+1}\mathcal K=\tilde a_{n+1},\quad I^B_{n+1}\mathcal K=a_{n+1},\\
		&\mathcal B_{}-\mathcal K-(\mathcal A_{}-\textit{I}_{3\times3})\mathcal K=0_{3\times1},\\
		&\mathcal B_{}+\overline{\mathcal B}_{}-\mathcal K-(\mathcal A_{}+\overline{\mathcal A}_{}-\textit{I}_{3\times3})\mathcal K=0_{3\times1},
	\end{split}\right.
	\end{equation}
where $0_{3\times 1}$ is $3\times 1$ vector with zero entries. 
Hence we have that
\begin{equation}\label{eq:xi-n+1-star}
	\begin{split}
		&~\xi^*_{n+1}\Big(\mathcal X_{(n+1)\Delta-}\Big)\\
	    =&~\xi^*_{n+1}(\mathcal X_{n\Delta+}-\mathcal K\xi^*_n)-\xi^*_{n}(\mathcal X_{n\Delta+}-\mathcal K\xi^*_n)+O(\Delta^2)\\
	    &~-\frac{I^A_{n+1}}{\tilde a_{n+1}}(\mathcal A_{}-\textit{I}_{3\times3})\left(\mathcal X_{n\Delta+}-\mathbb E[\mathcal X_{n\Delta+}]\right)-\frac{I^B_{n+1}}{a_{n+1}}(\mathcal A_{}+\overline{\mathcal A}_{}-\textit{I}_{3\times3})\mathbb E[\mathcal X_{n\Delta+}]-\frac{I^B_{n+1}}{a_{n+1}}\mathcal C_{}\\
	    &~-\frac{I^A_{n+1}}{\tilde a_{n+1}}\mathcal D_{}\epsilon_{n+1}\\
		:=&~\mathbb{I}+\mathbb{II}+\mathbb{III},
	\end{split}
\end{equation}
where we use $\mathcal X_{n\Delta+}$ to denote the state at $n\Delta$ after $\xi^*_n$ is implemented. 
Let $\widetilde Z$ be the candidate optimal strategy, and let $\delta\widetilde Z_{(n+1)\Delta} := \xi^*_{n+1}(\mathcal X_{(n+1)\Delta-})$. \eqref{eq:xi-n+1-star} can be written as 
\[
	\delta\widetilde Z_{(n+1)\Delta} =  \Big( \frac{\mathbb I + \mathbb{II}} {\Delta} \Big) \Delta - \frac{I^A_{n+1}}{\tilde a_{n+1}}\mathcal D_{} \epsilon_{n+1}. 
\]
This suggests to study the limit of $\frac{\mathbb{I}}{\Delta}$ and limit of $\frac{\mathbb{II}}{\Delta}$ as $\Delta \to 0$. By \eqref{eq:xi*n}, we have that
\begin{equation}\label{prelimit-deltaxi}
	\begin{split}
	\frac{\mathbb{I}}{\Delta}=&\frac{\xi^*_{n+1}(\mathcal X_{n\Delta+}-\mathcal K\xi^*_n)-\xi^*_{n}(\mathcal X_{n\Delta+}-\mathcal K\xi^*_n)}{\Delta}+\frac{O(\Delta^2)}{\Delta}\\
	=&-\frac{I^A_{n+1}-I^A_{n}}{\Delta\tilde a_{n+1}}\left(\mathcal X_{n\Delta+}-\mathcal K\xi^*_n-\mathbb E[\mathcal X_{n\Delta+}-\mathcal K\xi^*_n]\right) \\
	& -\frac{I^A_{n}}{\Delta}\left(\frac{1}{\tilde a_{n+1}}-\frac{1}{\tilde a_{n}}\right)\left(\mathcal X_{n\Delta+}-\mathcal K\xi^*_n-\mathbb E[\mathcal X_{n\Delta+}-\mathcal K\xi^*_n]\right)\\
	&-\frac{I^B_{n+1}-I^B_{n}}{\Delta a_{n+1}}\mathbb E[\mathcal X_{n\Delta+}-\mathcal K\xi^*_n]-\frac{I^B_{n}}{\Delta}\left(\frac{1}{a_{n+1}}-\frac{1}{a_{n}}\right)\mathbb E[\mathcal X_{n\Delta+}-\mathcal K\xi^*_n]\\
	&-\frac{I^D_{n+1}-I^D_{n}}{\Delta a_{n+1}}-\frac{I^D_{n}}{\Delta}\left(\frac{1}{a_{n+1}}-\frac{1}{a_{n}}\right)+\frac{O(\Delta^2)}{\Delta}.
	\end{split}
\end{equation}

Since $\tilde a_{n+1},\tilde a_{n}, a_{n+1}, a_{n}=O(\Delta),$ we get that
\begin{equation*}
	\begin{split}
	&\frac{\tilde a_{n+1}-\tilde a_n}{\Delta^3}=\rho^2\frac{A_{11,n+1}-A_{11,n}}{\Delta}\rightarrow\rho^2\frac{\dot A_{11,t}}{dt}\\
	&\frac{ a_{n+1}- a_n}{\Delta^3}=\frac{(\gamma_1\alpha-\gamma_2\rho)^2}{\gamma_2^2}\frac{B_{11,n+1}-B_{11,n}}{\Delta}+\frac{2\alpha(\gamma_1\alpha-\gamma_2\rho)}{\gamma_2}\frac{B_{13,n+1}-B_{13,n}}{\Delta}+\alpha^2\frac{B_{33,n+1}-B_{33,n}}{\Delta}	\\
	&\qquad\qquad\quad\rightarrow\frac{(\gamma_1\alpha-\gamma_2\rho)^2}{\gamma_2^2}\dot B_{11,t} +\frac{2\alpha(\gamma_1\alpha-\gamma_2\rho)}{\gamma_2}\dot B_{13,t}+\alpha^2\dot B_{33,t},
	\end{split}
\end{equation*}
from which we deduce that 
\begin{equation*}
	\begin{split}
		&\frac{1}{\tilde a_{n+1}}-\frac{1}{\tilde a_{n}}=-\frac{\tilde a_{n+1}-\tilde a_{n}}{\tilde a_{n+1}\tilde a_{n}}=-\frac{\Delta^2}{\tilde a_{n+1}\tilde a_{n}}\cdot\frac{\tilde a_{n+1}-\tilde a_{n}}{\Delta^3}\Delta\rightarrow 0,\\
		&\frac{1}{a_{n+1}}-\frac{1}{a_{n}}=-\frac{a_{n+1}-a_{n}}{a_{n+1}a_{n}}=-\frac{\Delta^2}{a_{n+1}a_{n}}\cdot\frac{a_{n+1}-a_{n}}{\Delta^3}\Delta\rightarrow 0.
	\end{split}
\end{equation*}
Furthermore, we have that
\begin{equation*}
	\begin{split}
		&\frac{(I^A_{n+1} -I^A_{n})\mathcal K}{\Delta^3}=\rho^2\frac{A_{11,n+1}-A_{11,n}}{\Delta}\rightarrow\rho^2\dot A_{11,t}\\
		&\frac{(I^B_{n+1} -I^B_{n})\mathcal K}{\Delta^3}=\frac{(\gamma_1\alpha-\gamma_2\rho)^2}{\gamma_2^2}\frac{B_{11,n+1}-B_{11,n}}{\Delta}+\frac{2\alpha(\gamma_1\alpha-\gamma_2\rho)}{\gamma_2}\frac{B_{13,n+1}-B_{13,n}}{\Delta}+\alpha^2\frac{B_{33,n+1}-B_{33,n}}{\Delta}	\\
		&\qquad\qquad\qquad\quad\rightarrow\frac{(\gamma_1\alpha-\gamma_2\rho)^2}{\gamma_2^2}\dot B_{11,t} +\frac{2\alpha(\gamma_1\alpha-\gamma_2\rho)}{\gamma_2}\dot B_{13,t}+\alpha^2\dot B_{33,t}.
	\end{split}
\end{equation*}
Altogether, we conclude that
\begin{equation*}\label{convergence-I1/Delta}
	\begin{split}
		\frac{\mathbb{I}}{\Delta}\rightarrow&~  -\frac{\dot I^A_t}{\tilde a}({\mathcal X_t}-\mathbb E[{\mathcal X_t}])-\frac{\dot I^B_t}{a}\mathbb E[{\mathcal X_t}]-\frac{\dot I^D_t}{a}.
	\end{split}
\end{equation*}
The limit of $\frac{\mathbb{II}}{\Delta}$ as $\Delta \to 0$ reads 
\begin{equation*}\label{convergence-I2/Delta}
	\begin{split}
		\frac{\mathbb{II}}{\Delta}=&~\frac{1}{\Delta}\bigg(-\frac{I^A_{n+1}}{\tilde a_{n+1}}(\mathcal A_{}-\textit{I}_{3\times3})\left(\mathcal X_{n\Delta+}-\mathbb E[\mathcal X_{n\Delta+}]\right)\\
		&-\frac{I^B_{n+1}}{a_{n+1}}(\mathcal A_{}+\overline{\mathcal A}_{}-\textit{I}_{3\times3})\mathbb E[\mathcal X_{n\Delta+}]-\frac{I^B_{n+1}}{a_{n+1}}\mathcal C_{}\bigg) \\	
		\rightarrow	&-\frac{I^A_t}{\tilde a}\mathcal H({\mathcal X_t}-\mathbb E[{\mathcal X_t}])-\frac{I^B_t}{a}\Big((\mathcal H+\overline{\mathcal H})\mathbb E[{\mathcal X}_t]+\mathcal G\Big).
	\end{split}
\end{equation*}  
Moreover, heuristically,  
\begin{equation*}\label{convergence-I3/Delta}
	\begin{split}
			~\mathbb{III}=&-\frac{I^A_{n+1}}{\tilde a_{n+1}}\mathcal D_{}\epsilon_{n+1}\rightarrow -\frac{I^A_t}{\tilde a}\mathcal D_t\,dW_t.
	\end{split}
\end{equation*}
Combining the above limits we obtain the candidate trading strategy on $(t,T)$:
\begin{equation*}\label{candidate}
	\begin{split}
	-dX_s = d\widetilde Z_s=&\left(-\frac{\dot I^A_s}{\tilde a}({\mathcal X_s}-\mathbb E[{\mathcal X_s}])-\frac{\dot I^B_s}{a}\mathbb E[{\mathcal X_s}]-\frac{\dot I^D_s}{a}-\frac{I^A_s}{\tilde a}\mathcal H({\mathcal X_s}-\mathbb E[{\mathcal X_s}]) \right. \\ 
	& ~ \left. - \frac{I^B_s}{a}((\mathcal H+\overline{\mathcal H})\mathbb E[{\mathcal X}_s]+\mathcal G)\right)\,ds -\frac{I^A_s}{\tilde a}\mathcal D_s\,dW_s. 
	\end{split}
\end{equation*}

\end{appendix}
\normalem 
\bibliography{bib_FHX2021}

\end{document}